\documentclass[12pt]{article}
\usepackage{epsfig, amssymb, graphicx, amsmath}
\usepackage{amsthm}
\usepackage{relsize}
\usepackage{comment}
\usepackage{mathcomp}

\usepackage[round-mode=places, round-integer-to-decimal, round-precision=2,
table-format = 1.2,
table-number-alignment=center,
round-integer-to-decimal,
output-decimal-marker={,}
]{siunitx}
\usepackage{booktabs}
\usepackage[shortlabels]{enumitem}
\usepackage{eurosym}	% simbolo dell'euro
\usepackage{wrapfig}	% immagini con testo intorno
\usepackage{float}		% immagini affiancate
\usepackage{footnote}

\usepackage{caption}
\usepackage{subcaption}
\usepackage{xcolor}
\usepackage{physics}
\usepackage{cases}
\usepackage{setspace}

\usepackage[most]{tcolorbox}
\tcbuselibrary{listings,skins}
\usepackage{algorithm}  % algoritmi simulazione
\usepackage[noend]{algpseudocode} % algoritmi simulazione

\usepackage{tikz}
\usetikzlibrary{patterns}
\usetikzlibrary{patterns.meta}
\usepackage{pgfplots}
\usetikzlibrary{shapes.geometric}

\usepackage{thmtools}
\usepackage{bbm}		% Per indicatore
\usepackage[title]{appendix}
\usepackage{adjustbox}
\usepackage{multirow}
\usepackage{hyperref}
\usepackage{multicol}	% pacchetto per scrivere su più colonne
\usepackage[authoryear]{natbib}
\bibpunct{(}{)}{,}{a}{}{,}

\defcitealias{solvencyII}{Solvency II Directive, 2009}
\defcitealias{baselII}{BCBS II 2005}
\defcitealias{baselIII}{BCBS III 2010}

\newtheorem{theorem}{Theorem}

\newtheorem{proposition}[theorem]{Proposition}
\newtheorem{lemma}[theorem]{Lemma}
\newtheorem{corollary}[theorem]{Corollary}
\theoremstyle{definition}

\theoremstyle{remark}

\newtheorem*{example*}{Example}
\numberwithin{theorem}{section}
\numberwithin{proposition}{section}
\numberwithin{lemma}{section}
\numberwithin{corollary}{section}
\numberwithin{definition}{section}
\newcommand{\be}{\begin{equation}}
    \newcommand{\en}{\end{equation}}
\newcommand{\ben}{\begin{equation*}}
	\newcommand{\enn}{\end{equation*}}
\newcommand{\bea}{\begin{eqnarray}}
	\newcommand{\ena}{\end{eqnarray}}

\newcommand{\lcp}{\mathrm{L}_{\boldsymbol{\lambda}}}
\newcommand{\llcp}{\mathcal{L}_{\boldsymbol{\lambda}}}
\newcommand{\ellcp}{\ell_{\boldsymbol{p}}}
\newcommand{\lcb}{\mathrm{L}_{\boldsymbol{p}}}
\newcommand{\llcb}{\mathcal{L}_{\boldsymbol{p}}}
\newcommand{\llcbsp}{\mathcal{L}}
\newcommand{\rhocp}{\widehat{\boldsymbol{\varrho}}_{\boldsymbol{\lambda}}}
\newcommand{\rhocb}{\widehat{\boldsymbol{\varrho}}_{\boldsymbol{p}}}

\DeclareMathOperator*{\argmax}{arg\,max}

\newcolumntype{C}[1]{>{\centering\let\newline\\\arraybackslash\hspace{0pt}}m{#1}}

\makeatletter
\newcommand*\bigcdot{\mathpalette\bigcdot@{.5}}
\newcommand*\bigcdot@[2]{\,\mathbin{\vcenter{\hbox{\scalebox{#2}{$\m@th#1\bullet$}}}}\,}
\makeatother

\usepackage[most]{tcolorbox}
\tcbuselibrary{listings,skins}

\algrenewcommand\algorithmicdo{}
\algrenewcommand\algorithmicthen{}

\begin{document}
	
\newlength\tindent
\setlength{\tindent}{\parindent}
\setlength{\parindent}{1em}
\renewcommand{\indent}{\hspace*{\tindent}}
	
\begin{savenotes}
\title{\bf{
    %Comb-Bernoulli and an application in insurance
    %modelling Dependence in Sparse Time Series
    Modeling dependence in sparse time series \\ of Insurance Claims
    }}
    \author{
        Roberto Baviera$^{1}$,
        Pietro Manzoni$^{2}$, \&
        Michele Domenico Massaria$^{1}$
    }
		
    \maketitle
    
    \begin{tabular}{ll}
    &  $^1$ Politecnico di Milano, Department of Mathematics, Italy \\
    &  $^2$ University of Edinburgh, Business School, United Kingdom \\
    \end{tabular}
\end{savenotes}

%\vspace{0.5cm}
%\begin{center}
%(Preliminary version, do not distribute)
%\end{center}
\vspace*{0.11truein}
\begin{abstract}
    \noindent
    Modeling the dependence between multiple risk types is a central challenge in contemporary insurance risk management. The standard approaches, 
    %such as L\'evy copulas, common shock models, and zero-mixed models,
    L\'evy copulas and zero-mixed models,
    often face practical difficulties in simulation and parameter calibration.
    In this paper, we introduce the Comb-Bernoulli model, a novel framework for capturing dependence between sparse time series of insurance risks, bridging the benefits of the two standard approaches.
    The (traditional) copula structure of the proposed model enables tractable: \textit{i)} simulation, 
    \textit{ii)} likelihood evaluation, and
    \textit{iii)} estimation of dependence parameters.
    We present the general properties of the model and analyze in detail the Gaussian copula case with lognormal marginals. Moreover, we illustrate an application using the Danish fire insurance dataset, highlighting both the modeling strengths and numerical efficiency of our approach in real-world risk management.
\end{abstract}
    
	\vspace*{0.11truein}
	\vspace*{0.11truein}

		\vspace{1.75cm}
		\noindent
		{\bf Address for correspondence:}\\
		{\bf Roberto Baviera}\\
		Department of Mathematics \\
		Politecnico di Milano\\
		32 p.zza Leonardo da Vinci \\
		I-20133 Milano, Italy \\
		Tel. +39-02-2399 4575\\
		roberto.baviera@polimi.it

\newpage

\setstretch{1.2}

\section{Introduction}\label{sec:introduction}

Accurate modeling of dependence in claims time series is crucial for insurance 
risk management. Dependence affects both occurrence (claim frequency) and severity (claim amount), directly 
impacting pricing decisions, reserve calculations, and regulatory capital 
requirements under frameworks such as Solvency II 
\citep[see, e.g.,][]{embrechts2002correlation,McNeil2005}.

%Dependence in claims’ time series can have an impact on both occurrence (claim frequency) and magnitude (claim amount). It has direct implications on several key functions within an insurance company, such as pricing, reserving, and Solvency II risk-based regulatory capital \citep{embrechts2002correlation,McNeil2005}.

In the insurance sector, claims data are time series often characterized by long stretches of zero values --~periods without observed claims~-- interspersed with infrequent but significant non-zero events.
The main challenge in modeling such time series is to simultaneously capture both the intermittent nature of claims and the dependencies across different risks, all in a single multivariate model that can accommodate both zero and non-zero values.
%modelling such sparse time series involves the main challenge of describing the intermittent nature of claims, capturing dependencies across different risk types in the joint distribution of both zero and non-zero values. 
It is well known and documented that the presence of probability atoms in financial time series may produce non-obvious behaviors already in the one-dimensional case \citep[see, e.g.,][]{bandi2020zeros,kolokolov2024jumps}.

\smallskip

In the literature, the standard modeling approach for such data is
%The standard approach 
%found in the literature for such data is
%The standard approach in the literature is 
%offered by 
a copula-based mixed distribution model
with a probability mass at zero (hereinafter \textit{zero-mixed model}), 
which \textit{i)} represents each marginal distribution as a combination of an atom
at zero and a continuous distribution for positive values,
and \textit{ii)} couples \textit{only} the positive marginals using a (traditional) copula.
This technique is common in domains such as hydrology and meteorology, where measurements 
(e.g.\ daily rainfall) are frequently zero.
It was introduced in the seminal paper of \cite{shimizu1993bivariate}, who first employed it %
%in meteorology using a copula-based construction 
to capture joint behaviors in sparse 
%positive-valued 
meteorological time series,
and has been considered by other researchers in the same field \citep[see, e.g.,][]{herr2005generic,zhang2007bivariate,serinaldi2009copula}.
In this framework, events can occur in two different locations at the same time, and the dependence between their magnitudes is modeled via a copula.
Crucially, the dependence applies exclusively to simultaneous non-zero events.\footnote{
The continuous-time generalization of the zero-mixed model is known in the insurance sector as common shock model, where a dependence structure is introduced between multiple compound Poisson processes that describe different classes of risk %(Yuen and Wang (2002), 
\citep[see, e.g.,][]{lindskog2003common}. In practice, these classes share ‘common shocks’ --~claims occurring exactly at the same time in two or more classes according to an identical arrival process.
}

Overall, the zero-mixed approach offers two key advantages. First, its simplicity allows for direct modelling of claims time series, effectively distinguishing between zero and non-zero regimes. Second, due to its flexibility, the model provides detailed and separate specifications for the dependence structures in both the occurrence and the size of events.

The main limitation lies in the estimation of model parameters, even for two-dimensional time series. In many relevant cases, co-events are rare, making it difficult to calibrate the dependence accurately. Additionally, the zero-mixed model becomes increasingly parameter-intensive, with the number of parameters growing exponentially with the dimensionality; for instance, a four-dimensional time series requires specifying sixteen probabilities (to account for events shared across two, three, or all four classes), six bivariate copulas, four trivariate copulas, one quadrivariate copula, and four jump size distributions, as already observed by \cite{avanzi2011modelling}. Finally, simulation of the zero-mixed model must proceed in two stages --~first determining the type of event, then simulating the sizes~-- which limits scalability beyond two dimensions.

%\smallskip
%In this paper, our first objective is to introduce a model that overcomes these limitations, allowing both a parsimonious description and a straightforward simulation, maintaining the model simple; our second objective is to understand the conditions under which the model can be expected to perform well in estimation, even in the presence of few co-events.

\medskip
A continuous-time alternative to the zero-mixed approach is the Lévy copula, a method where dependence is introduced via a multivariate function --~the L\'evy copula~-- that couples the marginal tail integrals of the univariate L\'evy measures. It was introduced in the seminal paper of \cite{tankov2003dependence} for spectrally positive multidimensional L\'evy processes, and is thoroughly discussed in reference
%excellent 
textbooks \citep[see, e.g.,][]{Cont}. 
%In continuous-time, an alternative to the common shock representation discussed above is the L\'evy copula approach, a method where dependence is introduced via a multivariate function (the L\'evy copula) that couples the marginal tail integrals of the univariate L\'evy measures. It was introduced in the seminal paper of \cite{tankov2003dependence} for spectrally positive multidimensional L\'evy processes, and it can be found described in excellent textbooks \citep[see, e.g.,][]{Cont}. 
In the insurance sector, this technique was first applied by \citet{bregman2005ruin} to a bivariate compound Poisson process to estimate ruin probabilities for an insurance company with two classes of risks coupled via a L\'evy copula.
A similar approach has been used to model operational losses by \citet{bocker2008modelling, bocker2010multivariate} and \cite{biagini2009asymptotics}. \cite{avanzi2011modelling} discuss the application of a generic L\'evy copula. \cite{esmaeili2010} derive the explicit form of the bivariate likelihood 
--~in particular for the Clayton L\'evy copula in the compound Poisson case~-- enabling a numerical maximum likelihood estimation of the L\'evy copula parameters.

The main advantage of this model lies in the parsimony of L\'evy copulas. Although typically applied to two risk classes in the literature, they can be scaled to multiple risk classes without causing exponential growth in the number of parameters. %Additionally, the likelihood has an explicit form allowing for a direct parameter estimation.
Unfortunately, simulating $d$ risks coupled via a L\'evy copula remains challenging, as shown by \cite{esmaeili2010}: the algorithm requires first simulating $2^d-1$ independent Poisson processes, followed by generating the corresponding jump sizes associated with each process.

\medskip
In this paper, we introduce the Comb-Bernoulli, a model that offers a parsimonious description of dependence in discrete time series, and bridges the benefits of the zero-mixed model (that offers simplicity, operates in discrete-time, and is well-suited to observed claim time series) and the L\'evy copula approach (known for parsimony and scalability in higher dimensions).

%\smallskip
%The focus of this paper is a parsimonious description of dependence in discrete time series, which bridges the benefits of the zero-mixed model (that offers simplicity, operates in discrete-time, and is well-suited to observed claim time series) and the L\'evy copula approach (known for parsimony and scalability in higher dimensions).

The proposed model is based on (traditional) copulas, which jointly capture the dependence in both claim occurrence and magnitude: this leads to a tractable structure that significantly simplifies both model estimation and simulation. Model estimation is facilitated by an explicit likelihood function, which can be efficiently maximized using the Inference Functions for Margins (IFM) approach \citep{godambe1991estimating,joe1997multivariate}, yielding asymptotic confidence intervals (CIs) for the estimated parameters.
Moreover, as we demonstrate, the simulation of a time series involving $d$ classes requires generating only $d$ processes linked via a copula; this approach overcomes the scalability limitations of zero-mixed models and avoids the exponential growth in $d$ of L\'evy copula models \citep[see, e.g.,][]{esmaeili2010}.

\medskip
There are three main contributions of this paper. First, we propose a new parsimonious model to describe the dependence of sparse time series that allows for both a straightforward estimation and a simple simulation. 
%Second, we discuss in detail an example with the most common dependence model used in the insurance sector, the Gaussian copula. 
Second, we provide a detailed specification of the proposed model within the Gaussian copula framework -- the most widely used dependence structure in the insurance literature.
Finally, we illustrate the performance and the effectiveness of our approach using a well-known time series of insurance claims.

\smallskip
The remainder of the paper is organized as follows. In Section \ref{sec:model}, we introduce our new modeling framework, together with all details required for its practical implementation.
In Section \ref{sec:gaussian_copula}, we establish the main properties of the model under the Gaussian copula specification. 
In Section \ref{sec:Results}, we present the results of the estimation of our model using the Danish fire insurance dataset. 
In Section \ref{sec:conclusions}, we provide concluding remarks.

\section{The Model}
\label{sec:model}

%In this section, we introduce the Comb-Bernoulli model, our novel multivariate model 
%in which each component features a point mass at zero and a continuous distribution on the positive real line, while dependencies are structured via a copula.
%We first formalize the marginal and joint specifications, then derive the corresponding likelihood function, and finally present an algorithm for simulating from the model.

%In these domains, sparsity is not an anomaly but a defining feature of the data-generating process that must be accounted for when measuring dependence.

In this Section, we introduce the Comb-Bernoulli,\footnote{To clarify the name of the model, we draw an analogy with functional analysis: a sequence of discrete jumps in time can be represented as a sum of step functions, whose continuous-time derivative is the Dirac comb (i.e.,\ a sum of Dirac delta functions). In the discrete setting, an analogous “comb” is obtained, describing instantaneous changes occurring at predetermined times.}
our novel multivariate model for insurance claims data.
%our flexible multivariate framework for insurance claims data.
%modelling discrete-time insurance claims 
%Each marginal component features a point mass at zero and a continuous distribution on $(0, \infty)$, while a copula structures their interdependence.
We first present the model formulation, then derive the corresponding likelihood function, develop an algorithm for efficient simulation, and finally discuss parameter estimation.

\subsection{Model Description} \label{sec:modeldescription}

Let $\boldsymbol{X}=[X_1,\dots, X_d]^\top$ denote the vector of total claims across $d$ distinct risk classes over the period of interest. 
As standard in the actuarial literature \citep[see, e.g.][]{klugman2012loss}, we consider claims as non-negative quantities (i.e.\ as claims to the insurance company) taking value zero when no claim occurs, or strictly positive values otherwise.

As described in Section \ref{sec:introduction}, the Comb-Bernoulli features two main elements: 
\textit{i)} mixed discrete-continuous marginal distributions for each risk class 
--~combining a point mass at zero with a distribution supported on the positive real line;
\textit{ii)} a copula function that captures the dependence structure among risk classes.

%In this Section, we introduce the model adopted for the vector $\boldsymbol{X}=(X_1,..., X_d)^\top$, where each $X_i$ models a different risk class. We first present the marginal distributions for the rvs $X_i$ ($i=1,...,d$) and then we introduce the multivariate model.

\medskip
%To account for the inherent sparsity of insurance claims, each marginal random variable (hereinafter rv) $X_i$ (with $i=1,\dots,d$) is modeled such that the event $\{X_i = 0\}$ occurs with probability $1-p_i$, while the positive claim follows a continuous severity distribution. 
%
%is modeled as a mixture of a point mass at zero and a continuous severity distribution on $(0, \infty)$.
To account for the inherent sparsity of insurance claims, each component $X_i$ (with $i=1,\dots,d$) of the claim vector $\boldsymbol{X}$ is modeled as a random variable (hereinafter rv) that equals zero with probability $1-p_i$, and follows a continuous severity distribution over positive values with probability $p_i$.
The cumulative distribution function (cdf) $F_i(\bigcdot;p_i)$ and the probability density function (pdf) $f_i(\bigcdot;p_i)$ of the rv\ $X_i$ read
\begin{equation}
    \left\{
        \begin{aligned}
        F_i(x;p_i) &= (1-p_i) \,\mathbbm{1}_{\{x\geq0\}} + p_i \,\Psi_i ( x ;\boldsymbol{\theta}_i) \, \mathbbm{1}_{\{x>0\}} \,, \\[2mm]
        f_i(x;p_i) &= (1-p_i) \,\delta_{0}(x) + p_i \,\psi_i(x;\boldsymbol{\theta}_i) \, \mathbbm{1}_{\{x>0\}} \,,
        \end{aligned}
    \right.
\label{eq:MarginalModel}
\end{equation}
where $p_i \in (0,1]$ is the probability of a positive claim for the $i$-th component,
$\delta_{0}$ is the Dirac delta at zero,
$\psi_i(\bigcdot;\boldsymbol{\theta}_i)$ is a continuous pdf defined on $(0, \infty)$ governing the claim severity,
$\Psi_i(\bigcdot;\boldsymbol{\theta}_i)$ is the corresponding cdf,
and $\boldsymbol{\theta}_i$ is the vector of parameters for the 
severity distributions. %\footnote{The name Comb-Bernoulli originates from functional analysis: the 1-d pdf resembles a Dirac comb function.}

A natural example for the severity is the lognormal distribution, commonly used in the literature to model insurance claims \citep[see, e.g.,][]{McNeil2005, poudyal2021robust}. 
In this case, the distributional parameters are $\boldsymbol{\theta}_i = [\mu_i,\sigma_i]^\top$ and, for $x > 0$, we have
\begin{equation}
	\label{eq:lognormal_marginals}
	\left\{
	\begin{aligned}
	\psi_i(x; \boldsymbol{\theta}_i) &=
	\varphi \left( \frac{\ln x - \mu_i}{\sigma_i} \right) \,
        \frac{1}{x \sigma_i} \,,\\[2mm]
	\Psi_i(x; \boldsymbol{\theta}_i) &= \Phi \left( \frac{\ln x - \mu_i}{\sigma_i} \right) \,,
\end{aligned}
\right.
\end{equation}
with $\varphi(\bigcdot)$ and $\Phi(\bigcdot)$ denoting the pdf and cdf of the standard normal (st.n.) distribution.

%The joint behavior of the claim vector $\boldsymbol{X} = [X_1, \dots, X_d]^\top$ is modeled through a copula function. Let $C(\boldsymbol{u}; \boldsymbol{\varrho}): [0,1]^d \rightarrow [0,1]$ be an $d$-dimensional copula function, absolutely continuous with respect to the Lebesgue measure on $[0,1]^d$, and parametrized by a dependence parameter vector $\boldsymbol{\varrho}$. Its density is denoted by $c(\boldsymbol{u}; \boldsymbol{\varrho})$.
%

\medskip
To combine flexible marginal distributions with a tractable dependence structure, 
we construct the joint distribution of the claim vector $\boldsymbol{X} = [X_1, \dots, X_d]^\top$ through a traditional copula.
%Specifically, 
Let $C(\,\bigcdot\,; \boldsymbol{\varrho}): [0,1]^d \rightarrow [0,1]$ denote a $d$-dimensional copula function, absolutely continuous w.r.t.\ the Lebesgue measure on $[0,1]^d$ and parametrized by a dependence vector $\boldsymbol{\varrho}$;
moreover, let $c(\,\bigcdot\,; \boldsymbol{\varrho})$ denote its density.\footnote{For notational convenience, we shall use interchangeably $C(u_1,\dots,u_d;\boldsymbol{\varrho})$ and $C(\boldsymbol{u};\boldsymbol{\varrho})$, with $\boldsymbol{u} := [u_1,\dots,u_d]^\top$. The same convention applies to the copula density $c(\bigcdot;\boldsymbol{\varrho})$.}
The joint distribution of $\boldsymbol{X}$ is defined via Sklar’s Theorem \citep[see, e.g.,][Th.5.3, p.186]{McNeil2005}:
\begin{equation} \label{eq:copula_definition}
    \mathbb{P}(X_1 \le x_1, \dots, X_d \le x_d)
    = 
    C\big(F_1(x_1;p_1), \dots, F_d(x_d;p_d); \boldsymbol{\varrho}\big).
\end{equation}
Similarly, the joint survival function is given by the survival copula $\overline{C}$: 
%\citep[see, e.g.,][p.~141, eq.~(4.9)]{CLV}:
\begin{equation} \label{eq:complementary_copula}
    \mathbb{P}(X_1 > x_1, \dots, X_d > x_d)
    = \overline{C}\big(\overline{F}_1(x_1;p_1), \dots, \overline{F}_d(x_d;p_d); \boldsymbol{\varrho}\big),
\end{equation}
where $\overline{F}_i(\,\bigcdot\,;p_i) := 1-F_i(\,\bigcdot\,;p_i)$ denotes the survival function of the $i$-th marginal.

\smallskip

%Using the survival copula, we can compute the exact probability of observing any specific configuration of zero and non-zero claims. The following Lemma defines the probability that strictly positive claims occur only for a specific subset of marginals.
%
%\noindent
%This formula characterizes the discrete component of the multivariate distribution and explicitly defines the probability mass for any subset of simultaneous claims.

\medskip

The Comb-Bernoulli model is thus characterized by three distinct sets of parameters, each governing a specific aspect of the multivariate claim distribution:
the occurrence probabilities 
$\boldsymbol{p} := [p_1, \dots, p_d]^{\top}$, which determine the frequency of claim events in each risk class; the severity parameters
$\boldsymbol{\Theta}:=[\boldsymbol{\theta}_1, \dots, \boldsymbol{\theta}_d]$, which shape the marginal distributions for positive claims in each risk class; 
and the dependence parameter $\boldsymbol{\varrho}$, which controls the correlation structure between marginal distributions through the copula function.

This explicit decomposition into occurrence, severity, and dependence components constitutes the core analytical strength of the proposed framework.
As we discuss in the following sections, this modularity offers several advantages: it enhances interpretability in risk assessment by allowing practitioners to isolate the contribution of each component; it enables rapid and efficient simulation of multivariate claim scenarios due to the separable structure; and it naturally supports a sequential, multi-stage estimation of the parameters, improving numerical stability and computational efficiency.
%allowing each component to be calibrated efficiently and independently.

\subsection{Likelihood function}
Due to the discrete-continuous nature of the marginals, the likelihood for any observation $\boldsymbol{x} = [x_1, \dots, x_d]^\top$ must explicitly account for the atoms at zero. It is therefore useful to define the \emph{active set} $\mathcal{S}({\boldsymbol{x}})$ for an observation $\boldsymbol{x}$ as the indices of components with positive claims:
\[
    \mathcal{S}(\boldsymbol{x}) := \{\, i \in \{1, \dots, d\} : x_i > 0 \,\} \subseteq\mathbb{S} \,\, ,
\]
where $\mathbb{S}:=\{1,\dots,d\}$.
We denote its cardinality by $\mathsf{s}(\boldsymbol{x}) := |\mathcal{S}(\boldsymbol{x})|$.

An example can help illustrate this notation. If claims arise from five different risk classes (i.e.\ $d=5$) and the observed claim vector is 
$\boldsymbol{x} = [0,\, 12.5,\, 4.6,\, 0,\, 7]^\top$, 
then the active set is $\mathcal{S}(\boldsymbol{x}) = \{ 2,3, 5 \}$ and its cardinality is $\mathsf{s}(\boldsymbol{x}) = 3$.
%An example can clarify this notation. If the claims rise within five different types (i.e., $d=5$), and the observed claims (one observation of the time series) are $\boldsymbol{x} = [0,12.5,4.6, 0, 3]^\top$, the active set $\mathcal{S}(\boldsymbol{x}) = \{ 2,3, 5 \}$, its cardinality is $\mathsf{s}(\boldsymbol{x}) = 3$.
%The active set identifies the components that contribute a continuous part to the likelihood, while the remaining zero components correspond to the discrete mass at zero.

\smallskip
The active set partitions the components of each observation into two groups: those in $\mathcal{S}(\boldsymbol{x})$ contribute to the likelihood through a continuous density term, while those not in $\mathcal{S}(\boldsymbol{x})$ contribute only via the discrete probability mass 
at zero.
The resulting log-likelihood $\llcb$ for a time series of $N$ multivariate observations 
can then be stated as follows.\footnote{
%The subscript $\boldsymbol{p}$ in $\llcb$ highlights that --~unlike standard copula models with continuous marginals~-- the defining feature of the Comb-Bernoulli is the presence of atomic probabilities at zero, governed by the parameters $\{p_i\}_i$.}
The subscript $\boldsymbol{p}$ in $\llcb$ emphasizes that, unlike standard copula models with continuous marginals, the Comb-Bernoulli model features atomic probabilities at zero, controlled by the parameters $\boldsymbol{p} := [p_1, \dots, p_d]^{\top}$.
}

\begin{proposition}
\label{pr:CombBernoulli_ll} 
Let $\mathcal{X} := \{\boldsymbol{x}^{(t)}\}_{t=1}^N$ be a $d$-dimensional Comb-Bernoulli time series. The log-likelihood of the model is given by
%
%Let $\mathcal{X} := \{\boldsymbol{x}^{(t)}\}_{t=1}^N$ be a non-negative time series in $\mathbb{R}^{d}$. The log‑likelihood of a $d$-variate Comb‑Bernoulli model is given by
%
\begin{equation}
\label{eq:log_likelihood_CB}
~~
%\llcb(\boldsymbol{\varrho})
\mathcal{L}_{\boldsymbol{p}}\left( \boldsymbol{\Theta}, \boldsymbol{\varrho} \mid \mathcal{X}\right)
=
\sum_{\boldsymbol{x} \in \mathcal{X}} \,
\Biggl\{
\ln \Biggl[
\frac{\partial^{\, \mathsf{s}(\boldsymbol{x})} C(u_1, \dots, u_d; \boldsymbol{\varrho})}
{\prod_{i \in \mathcal{S}(\boldsymbol{x})} \partial u_i} 
\Biggr]_{u_i = F_i(x_i; p_i)}
+
\sum_{i \in \mathcal{S}(\boldsymbol{x})} 
\ln \big(\, p_i\psi_i(x_i; \boldsymbol{\theta}_i) \, \big)
\Biggr\} \,.
\end{equation}
\end{proposition}
\begin{proof}
    See Appendix \ref{sec:Proofs}
\end{proof}
%
%
\iffalse
\textcolor{black}{
\begin{proposition}
Let $\mathcal{X} = [\boldsymbol{x}^{(t)}]_{t=1}^N$ be a $d$-dimensional Comb-Bernoulli time series of length $N$. 
The log-likelihood of $\mathcal{X}$ is
\begin{equation}
\label{eq:ts_loglikelihood}
    \ell_{\boldsymbol{p}}(\boldsymbol{\Theta},\boldsymbol{\varrho}\mid\mathcal{X}) 
    = \sum_{t=1}^{N} 
    \left[ \,
    \sum_{i \in \mathcal{S}(\boldsymbol{x}^{(t)})} \ln \psi_i(x_i^{(t)};\boldsymbol{\theta}_i)
    + \ln 
    \left[
    \frac{\partial^{\,\mathsf{s}(\boldsymbol{x}^{(t)})} C(u_1,\dots,u_d;\boldsymbol{\varrho})}
    {\prod_{i \in \mathcal{S}(\boldsymbol{x}^{(t)})} \partial u_i} 
    \right]_{u_i = F_i(x_i^{(t)}; p_i)}
    \, \right]
    \,.
\end{equation}
\end{proposition}
\begin{proof}
    See Appendix \ref{sec:Proofs}
\end{proof}
}
\fi

\noindent
For example, a fully zero observation $\boldsymbol{x}=\boldsymbol{0}$ (i.e.\ such that $\mathcal{S}(\boldsymbol{x}) = \varnothing$) yields a particularly simple contribution to the log-likelihood.
Indeed, the marginal sum $\sum\limits_{i \in \mathcal{S}(\boldsymbol{x})} \ln(p_i \psi_i)$ is empty, and the partial derivative term reduces to the copula itself:
\begin{equation*}
\Biggl[
\frac{\partial^{\, \mathsf{s}(\boldsymbol{x})} C(u_1, \dots, u_d; \boldsymbol{\varrho})}
{\prod_{i \in \mathcal{S}(\boldsymbol{x})} \partial u_i}
\Biggr]_{u_i = F_i(x_i; p_i)}
\, = \,
C(\, 1-p_1, \dots, 1-p_d; \boldsymbol{\varrho} \,) \,.
\end{equation*}

\bigskip
The availability of a simple likelihood expression is one of the great advantages of the Comb-Bernoulli model, since it allows straightforward estimation of the dependence parameters $\boldsymbol{\varrho}$ via the IFM approach \citep[see, e.g.,][Ch.10]{joe1997multivariate}, as we discuss in Section \ref{ssec:estimation_of_the_copula_parameters}.
%Under this method the marginal parameters are first estimated from their respective univariate likelihoods; subsequently, the dependence parameters are obtained by maximizing the remaining copula likelihood. 

The bivariate ($d=2$) and trivariate ($d=3$) Comb-Bernoulli models provide the simplest illustrations of the proposed framework. 
Explicit log-likelihoods for these cases are presented
in Appendices~\ref{sec:biv_case} and \ref{sec:triv_case}.

\subsubsection{Likelihood in the continuous-time limit}
The Comb-Bernoulli model introduced in Section \ref{sec:modeldescription} operates in discrete-time, with claims 
%the aggregate claims are 
observed at $N$ monitoring times, equally spaced over the period $(0, T]$;
the time interval between two successive observations is $\Delta t:= T/N$.
At each monitoring time, the probability that the observed loss vector has $I \subseteq \mathbb{S}$ as active set, denoted as $p_I^{\perp}$\,, is given by the following Lemma.
\begin{lemma}\label{lem:prob_active_set}
    For any $I \subseteq \mathbb{S}$, the probability of observing positive losses exactly in the components indexed by $I$ is
    \begin{equation} \label{eq:p_I_perp}
    p_I^{\perp} 
    :=
    \mathbb{P}\big( \{X_i > 0 ~\, \forall i \in I\} \cap \{X_j = 0 ~\, \forall j \notin I\}\big) 
    =
    \sum_{J \supseteq I}	
    (-1)^{|J|-|I|} ~
	\overline{C}_{J}
	\left(
		\, \left\{ p_i \right\}_{i \in J} \,
	\right) \,\,,
    \end{equation}
    where the sum is taken over all index sets $J \subseteq \mathbb{S}$ containing $I$, and
    $\overline{C}_{J}$ denotes the survival copula restricted to the index set $J$ of cardinality $|J|$.\footnote{\label{fn:restriction_copula}As standard in the literature 
    \citep[see, e.g.,][]{CLV}, for any nonempty index set $J \subseteq \mathbb{S}$, we define the \textit{restriction} of the survival copula $\overline{C}$ to the coordinates in $J$ as
\[
    \overline{C}_J( \{ u_j \}_{j\in J})
    := 
    \overline{C}\bigl(
        \{ u_j \}_{j\in J}, \, \{ 1 \}_{j\in J^{\complement}} \, ; \,\boldsymbol{\varrho}\,
    \bigr)\,\,,
\]
i.e.\ as the survival copula obtained by setting all components not in $J$ as equal to $1$.}
\end{lemma}
\begin{proof}
    See Appendix \ref{sec:Proofs}
\end{proof}

%{ \color{black}
%\begin{example*}
As an example, in the bivariate case, the probability of observing a positive claim exclusively in the first risk class (i.e.\ with active set $I=\{1\}$) is given by
\[
p_{\{1\}}^{\perp} = \overline{C}_{\{1\}}(p_1) - \overline{C}_{\{1,2\}}(p_1, p_2) = p_1 - \overline{C}(p_1, p_2;\boldsymbol{\varrho}),
\]
which naturally translates to the difference $\mathbb{P}(X_1 > 0) - \mathbb{P}(X_1 > 0, X_2 > 0)$. Similarly, the probability of a simultaneous co-jump in both classes (i.e.\ $I=\{1,2\}$) reduces to $p_{\{1,2\}}^{\perp} = \overline{C}(p_1, p_2;\boldsymbol{\varrho})$.
%\end{example*}
%}

\medskip

In this Section, we show that the Comb-Bernoulli admits a natural continuous-time extension. 
We consider the limit model obtained by increasing the monitoring frequency (i.e.\ by letting $\Delta t \to 0^+$), while applying a Poisson scaling to the probabilities $p_I^\perp$\,:
\begin{equation}
    \phantom{\quad\quad \forall\,I\neq \varnothing}
    p_I^{\perp} \,=\,\lambda_I^{\perp}\,\Delta t
    \quad\quad 
    \forall\,I\neq \varnothing
    \,.
\end{equation}
This scaling ensures that, as the sampling time-step $\Delta t$ decreases, the expected number of claims with $I$ as active set over the period $(0, T]$ remains constant, equal to $\lambda_I^{\perp} T$.

In this continuous-time setting, claims can occur at any instant within the interval $(0, T]$, rather than only at the discrete monitoring times. Interarrival times follow exponential 
distributions characteristic of Poisson processes, and the intensities 
$\{\lambda_I^\perp\}_{I \subseteq \mathbb{S}}$ govern the occurrence of events with active set $I$. 
The dependence across risk classes is captured by a multivariate Lévy copula, connecting to the bivariate construction of \cite{esmaeili2010}.

The following Proposition presents the likelihood of this continuous-time representation. 
%The obtained expression encompasses the bivariate likelihood of \citet[Th.4.1, p.227]{esmaeili2010} 
%--~constructed using the different technique of two-dimensional Lévy copulas~-- 
%as a special case, while extending to $d$ dimensions.
%\smallskip
%We prove in the next Proposition the likelihood of this limit representation; it has been obtained for
%a bivariate compound Poisson with the different technique of two-dimensional L\'evy copulas in the work of \citet[][Th.4.1, p.227]{esmaeili2010}.

\begin{proposition} \label{pr:CombPoisson_ll}
Let $\mathcal{X} := \{\boldsymbol{x}^{(t)}\}_{t}$ denote a collection of observed claim events 
over $(0,T]$, where $t$ represents the arrival time and each $\boldsymbol{x}^{(t)}$ 
the associated $d$-variate loss vector.
In the continuous-time limit, the log-likelihood of the $d$-dimensional 
Comb-Bernoulli is given by
    \begin{equation}
    \label{eq:Comb_Poi_ll}
    \llcp\left(\boldsymbol{\Theta},\boldsymbol{\varrho}\mid\mathcal{X}\right) %\left(\boldsymbol{\varrho} \right)
    =
    \sum_{ I \neq \varnothing } ~
    \left[ ~\sum_{\substack{\boldsymbol{x} \in \mathcal{X} \,:\, \mathcal{S}(\boldsymbol{x}) = I }}
    \left( 
        \ln \zeta \bigl( \boldsymbol{x} \bigr) +
         \sum_{i \in I} \ln\left(
         \lambda_i \psi_i \bigl( x_i;\boldsymbol{\theta}_i \bigr)\right)
    \right) ~- T \lambda_I^{\perp} \,
    \right]
    \,.
\end{equation}
Here, $\lambda_i$ denotes the intensity of component $i$, defined as
$\lambda_i := \sum\limits_{I \subseteq \mathbb{S} \,:\, I \ni i} \lambda_I^{\perp}$, and the function $\zeta$ is

\begin{equation*}
\zeta (\boldsymbol{x})
:=
\sum_{J \supseteq \mathcal{S}(\boldsymbol{x})} (-1)^{|J|}
\left[
\frac{\partial^{\, \mathrm{s}(\boldsymbol{x})}}{\prod_{i \in \mathcal{S}(\boldsymbol{x})} \partial u_i}
\mathfrak{C}_J(\{ u_i \}_{i \in J})
\right]_{u_i = \lambda_i \overline{\Psi_i}(x_i)} ~,
\end{equation*}
%and
%\begin{equation*}
%    \lambda_I^\perp=\mathfrak{C}_{ I}\left(\{\lambda_i\}_{i\in I};\boldsymbol{\varrho}\right)\,\,,
%\end{equation*}
where $\mathfrak{C}_{J}$ denotes the L\'evy copula restricted to the index set $J$, i.e.\
\[
    \mathfrak{C}_{J} \bigl( \{u_i\}_{i \in J} \bigr) 
    := 
    \lim_{\Delta t \rightarrow 0}
    \frac{\overline{C}_J\bigl( \{\Delta t\,u_i\}_{i \in J} \bigr)}{\Delta t} \,\,.
\]
\end{proposition}
\begin{proof}
    See Appendix \ref{sec:Proofs}
\end{proof}

%As in the discrete-time case, parameter estimation can be efficiently performed using the IFM method, which exploits the likelihood factorization to estimate marginal parameters separately from dependence parameters.
%The algorithm for simulating this type of process has been developed in \cite{tankov2006simulation} and in \cite{esmaeili2010} for a bivariate process.

\noindent
This result establishes that the discrete-time Comb-Bernoulli model converges to a Lévy copula model in continuous-time.
Notably, the obtained likelihood expression \eqref{eq:Comb_Poi_ll} recovers the likelihood of the bivariate Lévy copula model by \citet[Th.~4.1, p.~227]{esmaeili2010} as a special case, while providing a straightforward generalization to arbitrary dimension $d$.

Proposition \ref{pr:CombPoisson_ll} thus formalizes a theoretical bridge between discrete- and continuous-time dependence modeling: the Comb-Bernoulli represents a tractable discrete-time analog of Lévy copula models, and its elementary structure offers substantial computational advantages in both likelihood evaluation and model simulation. The latter aspect is the focus of the next section.

%
%In particular, while the simulation of the continuous-time model can be performed using the existing Lévy copula algorithms
%\citet[][for the general multivariate case]{tankov2006simulation}
%and \cite[][in the bivariate setting]{esmaeili2010}.

\subsection{Simulation}
\label{sec:ParametricSimulation}

The Comb-Bernoulli model can be simulated efficiently using a straightforward two-step procedure. Algorithm \ref{al:parametric} presents the complete simulation scheme for generating a time series of length $N$.

%------------------------------------------------------------------
% ORIGINAL ALGORITHM
%------------------------------------------------------------------
\iffalse
\begin{algorithm}[htbp]
	\caption{Comb-Bernoulli simulation}\label{al:parametric}
	\begin{algorithmic}
					
		\Procedure{Parametric Simulation}{$N$}       %\Comment{This is a test}
		\State
		\For{$t = 1, \ldots,N$} %\Comment{repeat for each time step}
		 \item \State DRAW $\boldsymbol{u}^{(t)} := \left[u_i^{(t)}\right]_{i=1}^d$ from the copula $C(u_1,\dots,u_d;\boldsymbol{\varrho})$
		 \Comment{Step (i): copula}
	
		\EndFor
		\State
		\For{$i = 1,\ldots,d$} 											\Comment{Step (ii): compute claims}
			\If{ $u_i^{(t)} \leq 1-p_i $}
				\State COMPUTE $\quad x_i^{(t)}$ $\gets$ $0  \,,$ %\Comment{null claim}			
			\Else
				\State COMPUTE 
$\quad x_i^{(t)}$ $\gets$ $\displaystyle \Psi_i^{-1} \left( \frac{u_i^{(t)}-(1-p_i)}{p_i};\boldsymbol{\theta}_i \right) \,.$	
			\EndIf			
		\EndFor
		\EndProcedure
					
	\end{algorithmic}
\end{algorithm}
\fi

\begin{tcolorbox}[colback=gray!5, colframe=black!75,
                 title=Algorithm~\refstepcounter{algorithm}\label{al:parametric}\thealgorithm.
                 Comb-Bernoulli Simulation,
                  fonttitle=\bfseries, sharp corners, boxrule=0.8pt]
{\setlength{\baselineskip}{1.6em}
\begin{algorithmic}[1]
    %\Procedure{ParametricSimulation}{$N$}
        \Require \textit{N} (number of simulations)
                 %$\boldsymbol{p} = [p_1,\dots,p_d]$, 
                 %copula $C(u_1,\dots,u_d;\boldsymbol{\varrho})$, 
                 %$\Psi_i^{-1}$ for each $i$

        \For{$t = 1, \dots, N$}
            \State Draw 
            $\boldsymbol{u}^{(t)} := \big[u_1^{(t)},\dots,u_d^{(t)} \big]^{\top}$ from $C(\boldsymbol{u}^{(t)};\boldsymbol{\varrho})$
            \Comment{Step 1: simulate copula}
        \For{$i = 1, \dots, d$} \Comment{Step 2: compute claims}
            \If{$u_i^{(t)} \leq 1 - p_i$}
                \State Update\, $x_i^{(t)} \gets 0$
            \Else
                \State Update\,
                $x_i^{(t)} \gets \displaystyle \Psi_i^{-1}
                \Big(\,\frac{u_i^{(t)} - (1-p_i)}{p_i}; \, \boldsymbol{\theta}_i \,\Big)$
            \EndIf
        \EndFor
        \EndFor
    %\EndProcedure
\end{algorithmic}
}
\end{tcolorbox}

\smallskip
%Let us briefly comment on the simulation procedure.
Through Step 1, we sample a $d$-dimensional vector
$\boldsymbol{u}^{(t)}$ from the copula $C(\,\bigcdot\,;\boldsymbol{\varrho})$ with dependence parameter $\boldsymbol{\varrho}$ (see Figure \ref{figure:figura_0}a for the bivariate case).
Step 2 then applies the key transformation of our model: the $i$-th component incurs a loss if and only if $u_i^{(t)}$ exceeds the threshold $1-p_i$\,; otherwise, it is set to zero. This threshold mechanism translates the copula structure into
claim occurrence and severity, 
forming the cornerstone of the Comb-Bernoulli framework.

In the bivariate case, for example, Step 2 gives rise to four distinct cases, corresponding to the four regions in Figure \ref{figure:figura_0}b: 
both claims are positive (\textit{upper-right});
a positive claim on $x_1$, zero on $x_2$ (\textit{lower-right});
a positive claim on $x_2$, zero on $x_1$ (\textit{upper-left});
both claims are zero (\textit{lower-left}).

\medskip
Algorithm \ref{al:parametric} offers three main advantages compared to the simulation algorithm of L\'evy copulas: \textit{i)} it can be easily vectorized, \textit{ii)} it requires the inversion only of univariate cdfs, \textit{iii)} it scales linearly in the number of risk classes $d$. 
Lévy copula simulation \citep[cf.][]{tankov2006simulation, esmaeili2010}, in contrast, does not present the same advantages. In particular, it scales exponentially in the number of risk classes, as it requires the simulation of $2^d-1$ separate processes, as well as the inversion of multivariate cdfs. For instance, in the trivariate case, 
one must invert three bivariate cdfs and a trivariate cdf --~a task that is computationally onerous in most cases.

A detailed quantitative assessment of these computational gains is provided in Section \ref{sec:Results}.

\begin{figure}[t]
	\centering
	\scalebox{0.95}{
		\begin{tikzpicture}

	% colored areas
	%\fill[gray!40] (1.50, 1.10) rectangle (4.00, 4.00);

	\foreach \Point in {
	(3.77,3.51), 
	(2.42,2.33), 
	(0.29,0.91), 
	(3.33,1.34), 
	(3.88,3.84), 
	(3.87,3.20), 
	(2.40,1.39), 
	(0.66,0.53),
	(0.97,0.68), 
	(3.30,3.16), 
	(0.47,0.51), 
	(0.35,0.90), 
	(1.76,0.47),
	(0.78,2.29),
	(1.19,1.88),
	(2.51,2.51),
	(2.63,3.87),
	(2.66,2.87),
	(2.95,1.93),
	(2.45,1.27),
	(1.73,0.17),
	(1.47,3.09),
	(2.31,3.01),
	(1.87,3.69),
	(3.80,3.92),
	(3.22,1.84),
	(0.46,0.81),
	(0.57,1.75),
	(2.43,1.26),
	(3.13,3.25),
	(0.72,1.40),
	(2.17,2.31),
	(0.73,0.08),
	(1.41,2.58),
	(1.63,2.44),
	(1.67,2.22),
	(2.69,2.56),
	(1.74,2.63),
	(0.64,0.96),
	(3.71,3.06),
	(2.25,1.60),
	(0.17,1.59),
	(2.35,3.22),
	(3.58,3.13),
	(0.87,1.82),
	(1.12,2.12),
	(2.89,2.59),
	(0.56,1.14),
	(0.10,0.61),
	(2.15,2.13),
	(2.94,0.97),
	(3.41,1.87),
	(1.22,0.83),
	(3.19,1.28)
	}{
    	\node[blue] at \Point {\textbullet};
	}

	% axes
    \draw[->, line width=0.3mm] (-0.5, 0) -- (5.5, 0);
    \draw[->, line width=0.3mm] (0, -0.5) -- (0, 5.5);
    \node at (5.00, 0.30) {$u_1$};
    \node at (0.35, 5.00) {$u_2$};

	% other perimeter line
 	\draw[-, line width=0.5mm] (0, 4) -- (0, 0);
    \draw[-, line width=0.5mm] (4, 0) -- (0, 0);
    \draw[-, line width=0.5mm] (0, 4) -- (4, 4);
    \draw[-, line width=0.5mm] (4, 0) -- (4, 4);

    \node at (4.00, -0.35) {$1$};
    \node at (-0.25, 4.00) {$1$};
    
	% other perimeter line
    \draw[dashed, line width=0.3mm, color=gray!60] (0, 1.1) -- (4, 1.1);
    \draw[dashed, line width=0.3mm, color=gray!60] (1.5, 0) -- (1.5, 4);

    %\node at (1.50, -0.35) {$1-p_1$};
    %\node at (-0.85, 1.10) {$1-p_2$};
    
    \node at (2, -1.25) {(a)}; 
    
\end{tikzpicture}
% End of code
	}
	\hspace{0.5cm}
	\scalebox{0.95}{
		\begin{tikzpicture}

	% colored areas
	\fill[gray!40] (1.50, 1.10) rectangle (4.00, 4.00);

	\draw [pattern={Lines[angle=45,distance=5pt]},pattern color=gray!60] 
    (1.50, 0.00) rectangle (4.00, 1.10);
	\draw [pattern={Lines[angle=45,distance=5pt]},pattern color=gray!60] 
    (0.00, 1.10) rectangle (1.50, 4.00);

	% axes
    \draw[->, line width=0.3mm] (-0.5, 0) -- (5.5, 0);
    \draw[->, line width=0.3mm] (0, -0.5) -- (0, 5.5);
    \node at (5.00, 0.30) {$u_1$};
    \node at (0.35, 5.00) {$u_2$};

	% other perimeter line
 	\draw[-, line width=0.5mm] (0, 4) -- (0, 0);
    \draw[-, line width=0.5mm] (4, 0) -- (0, 0);
    \draw[-, line width=0.5mm] (0, 4) -- (4, 4);
    \draw[-, line width=0.5mm] (4, 0) -- (4, 4);

    \node at (4.00, -0.35) {$1$};
    \node at (-0.25, 4.00) {$1$};
    
	% other perimeter line
    \draw[-, line width=0.3mm] (0, 1.1) -- (4, 1.1);
    \draw[-, line width=0.3mm] (1.5, 0) -- (1.5, 4);

    \node at (1.50, -0.35) {$1-p_1$};
    \node at (-0.85, 1.10) {$1-p_2$};

	%\draw[pattern=crosshatch dots gray, pattern color=blue] (1.50, 1.10) rectangle (4.00, 4.00);
%	\draw[pattern=north west lines, pattern color=blue] (1.50, 1.10) rectangle (4.00, 4.00);

    \node at (2, -1.25) {(b)};     
    
\end{tikzpicture}
% End of code
	}
	\caption{\small Simulation of a two-dimensional Comb-Bernoulli. 
	First, a pair $(u_1, u_2)$ is sampled from a two-dimensional copula (\textit{left}), following Step 1 of Algorithm~\ref{al:parametric}.
	Next, the claims are generated by applying Step 2:
	four cases are possible, corresponding to the four depicted regions (\textit{right}). 
We indicate the co-jump region in grey, the no-jump region in white, and 
the region with only a non-zero claim ($x_1$ or $x_2$) with a hatch pattern.
	}
	\label{figure:figura_0}
\end{figure}
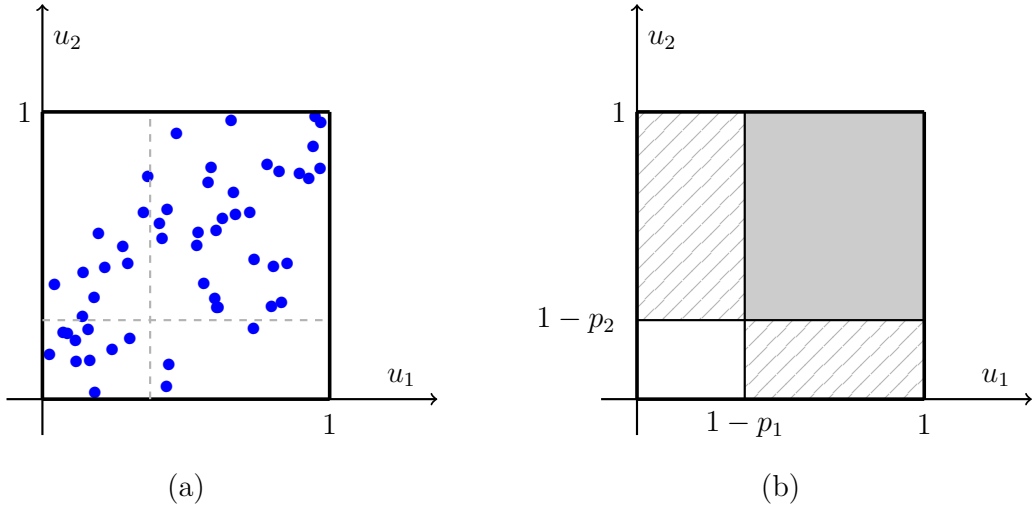

\subsection{Parameter estimation} \label{ssec:estimation_of_the_copula_parameters}

This section develops the estimation methodology for the parameters of the Comb-Bernoulli model, with particular focus on the copula dependence $\boldsymbol{\varrho}$. 
We first present a technique to obtain point estimates of the model parameters, based on the IFM method, then we discuss two methodologies for the construction of the confidence intervals.

%\footnote{The method described in this Section is valid for any choices of the copula $C(\bigcdot;\boldsymbol{\varrho})$, that admits a bounded (from above) likelihood.}

\subsubsection{Point estimation} \label{sec:param_estim}

The IFM method is a two-stage estimation procedure widely used in the copula literature \citep[see, e.g.,][and references therein]{CLV,cesari2019}. 
In finance, this method is an example of cascade estimation, where the parameters related to model parts requiring greater accuracy (e.g., more liquid financial products) are calibrated first, with the remaining parameters estimated subsequently
\citep[see, e.g.][]{brigo2005empirically,baviera2026additive}.

\smallskip\noindent
In the context of the Comb-Bernoulli model, the procedure is the following.

First, for each of the $d$ marginal distributions, one estimates the parameters $p_i$ and $\boldsymbol{\theta}_i$ using the univariate maximum likelihood estimators (MLEs) of the marginal density in \eqref{eq:MarginalModel}. %\footnote{
It is straightforward to show that, for the $i$-th marginal, the univariate MLE for $p_i$ is the fraction of positive observations in the $i$-th component, while the estimator for $\boldsymbol{\theta}_i$ is the MLE restricted to the positive observations 
--~effectively fitting $\psi_i$ to the non-zero losses only.
This yields closed-form estimators for the marginal parameters in all relevant cases, rendering the calibration extremely efficient and accurate. 
%}
We denote the obtained first-stage estimates as $\widehat{\boldsymbol{p}}$ and $\widehat{\boldsymbol{\Theta}}$.

Then, one estimates the copula dependence $\boldsymbol{\varrho}$ via the IFM estimator $\rhocb$, defined as
\begin{equation}        
\label{eq:argmax_p}
    \rhocb
    :=
    \argmax_{\displaystyle\boldsymbol{\varrho}
    }
    \mathcal{L}_{\widehat{\boldsymbol{p}}}\left(\widehat{\boldsymbol{\Theta}}, \boldsymbol{\varrho} \mid \mathcal{X}\right)\,,
\end{equation}
namely as the maximizer of the log-likelihood \eqref{eq:log_likelihood_CB} once the marginal parameters have been fixed at their first-stage estimates
$\widehat{\boldsymbol{p}}$ and $\widehat{\boldsymbol{\Theta}}$.

The maximization problem in \eqref{eq:argmax_p} can be solved numerically to compute the estimate $\rhocb$ for the copula parameter. While in practice this estimation can be performed by applying the first-order conditions \citep[see, e.g.,][]{joe1997multivariate}, the well-posedness of the estimator is valid under much broader assumptions. Specifically, by framing the proposed estimator $\rhocb$ within the broad class of extremum estimators, its consistency and asymptotic normality can be established for a general parametric copula, relying on the topological properties of the parameter space and the interiority of the true parameter, following the theoretical framework outlined in \cite{newey1994large}.

\subsubsection{Confidence Intervals}

A point estimate alone is of limited use in insurance applications, as risk assessment requires reliable quantification of the associated uncertainty.
In what follows,
%In the following,
we describe two main approaches to construct CIs for the Comb-Bernoulli parameters: a parametric bootstrap and a method based on the estimating functions theory.

\smallskip
The parametric bootstrap \citep[see, e.g.,][Ch.10]{efron2021computer} is a powerful numerical technique for construction of CIs. In the Comb-Bernoulli case, its implementation is elementary thanks to the availability of the efficient simulation algorithm developed in Section \ref{sec:ParametricSimulation}. 
The parametric bootstrap procedure involves three steps:
\begin{enumerate}[noitemsep, topsep=2pt]
\item Simulate $B$ \textit{replicas} of the original time series of claims, using Algorithm \ref{al:parametric} with the estimated parameters $\widehat{\boldsymbol{p}}$, $\widehat{\boldsymbol{\Theta}}$, and $\rhocb$;
\item For each replica, re-estimate the parameters of interest via IFM to obtain $B$ estimates;
\item Construct the CIs from the empirical quantiles at levels $\alpha/2$ and $1-\alpha/2$ of the bootstrapped estimates (e.g., $\alpha = 0.05$ for a $95\%$ CI). 
If multiple parameters are of interest, the Bonferroni correction \citep[see, e.g.,][]{dunn1958estimation} replaces $\alpha$ with $\alpha/m$, where $m$ is the number of parameters.
\end{enumerate}

For example, if we only aim to infer the dependence parameters $\boldsymbol{\varrho}$
--~here assumed to be an $m$-dimensional vector $\boldsymbol{\varrho} = [\varrho_1, \ldots, \varrho_m]^{\top}$~-- the bootstrapped $\alpha$-CI for each element $\varrho_j$ is
\[
\left[ \widehat{\varrho}_{j, \, \alpha/(2m)},\; \widehat{\varrho}_{j, \, 1-\alpha/(2m)} \right],
\]
where $\widehat{\varrho}_{j, \,q}$ denotes the empirical $q$-quantile of the bootstrapped estimates for $\varrho_j$.

\bigskip
As an alternative to the parametric bootstrap, asymptotic CIs for the Comb-Bernoulli parameters can be constructed using the estimating functions theory. Based on the results in \citet{joe1997multivariate} and \citet{newey1994large}, the distribution of the IFM estimates is asymptotically normal with covariance matrix given by the inverse of the Godambe information matrix -- a generalization of the Fisher information matrix that can be computed numerically 
\citep[see, e.g.,][p.~405, eq.~(2.6)]{joe2005asymptotic}. Individual CIs for the parameters of interest are then obtained either by projecting the joint (elliptical) confidence region onto the coordinate axes, or directly from the diagonal elements of the covariance matrix, applying a Bonferroni correction if needed.

\section{The Gaussian copula case}
\label{sec:gaussian_copula}

%[Cappello]
The Gaussian copula is widely used in the financial industry. Its popularity stems from several key advantages: considerable flexibility in specifying correlation parameters, straightforward implementation, and well-established theoretical properties \citep[see, e.g.,][]{embrechts2002correlation, guegan2012aggregation,zhang2013predicting}.

%[Banking]
The Gaussian copula has found extensive applications in the banking sector. The cornerstone of credit risk capital requirements in banks’ internal models 
-- \cite{gordy2003risk} -- is a Gaussian copula model with a separable correlation matrix. This specification forms the basis of the Basel II Internal Ratings-Based (IRB) framework \citepalias[][]{baselII} and its subsequent modifications \citepalias[see, e.g,][\citeauthor{hull2012risk} \citeyear{hull2012risk}]{baselIII}. This approach generalizes earlier seminal Gaussian copula credit models \citep{vasicek1987probability,li2000default}.

%[Insurance]
In insurance modeling, the Gaussian copula is used in determining solvency capital requirements, economic capital allocation, and diversification benefits, and in pricing certain multi-risk products. It can be readily scaled to the number of risk factors typically involved in Solvency II internal models (IM) and relies on relatively elementary formulas.
%
%[Integrated Risk management]
Moreover, in integrated enterprise-wide risk management, the Gaussian copula can play an important role within insurance companies. As emphasized by \citet[p.570]{rosenberg2006general}, ``the goal of integrated risk management is to both measure and manage risk and capital across a diverse range of risks in the insurance sector.'' Such an approach requires aggregating different types of risk (market, actuarial, and operational), whose underlying characteristics may vary considerably.

In their seminal paper, \cite{ward2002practical} were the first to derive the total loss distribution for a diversified insurer, using a normal copula to aggregate a heterogeneous set of risks. The Gaussian copula framework allows for the incorporation of realistic marginal distributions that capture key empirical features, while also providing a flexible dependence structure across risks.

On this respect, \citet[p.892]{wang1998aggregation} emphasizes that ``the normal copula is very flexible as it allows any (symmetric, positive definite) matrix of rank correlation coefficients (\ldots) In many practical situations, we only have partial information about correlation parameters without knowing the exact underlying multivariate distribution. In such cases, a normal copula provides a simple method for simulating correlated variables''. The Gaussian copula therefore satisfies two key requirements for integrated risk management: \textit{i)} it allows for a rich dependence structure across risks, and \textit{ii)} it remains applicable even when information about inter-risk dependencies is limited (e.g., when only correlations are available).

A Solvency II internal model may involve a large number of risk factors -- around $20$ when considering a single risk class (e.g., operational risk), and up to $100$ in a fully integrated model. Such high dimensionality poses challenges for the existing modeling approaches, as noted in 
%the Introduction. 
Section \ref{sec:introduction}.
Addressing this issue in the context of modeling dependence in sparse time series is a key objective of this work. In this setting, the Gaussian copula provides a natural and practical choice for both fitting and simulation in internal models.

Finally, it can be useful to emphasize that the results we present in this section can be extended to any elliptical copula. The Gaussian case remains a relevant benchmark, noting that, as highlighted by \citet[p.571]{rosenberg2006general}, the choice of elliptical copula (e.g., normal versus Student-t), which determines the degree of tail dependence (for continuous marginals), often has a relatively modest impact on total risk.

\smallskip
In this Section, we present the main properties of the Comb-Bernoulli model when the dependence structure is specified by a Gaussian copula.
First, we derive a closed-form expression for the log-likelihood, and then we analyze the asymptotic behavior of the model as the probability masses at zero vanish, demonstrating consistency with classical copula theory.

\subsection{Elementary closed-form likelihood expression}
%\subsection{Closed-form likelihood expression}

In the Gaussian copula, the dependence is fully encoded by a correlation matrix $\mathbf{R} \in \mathbb{R}^{d \times d}$.
%which is symmetric and positive definite.
The generic copula parameter $\boldsymbol{\varrho}$ in \eqref{eq:copula_definition} is therefore 
identified with this matrix, 
%i.e.\ $\boldsymbol{\varrho} \equiv \mathbf{R}$, 
with the understanding that $\boldsymbol{\varrho}$ is a vector consisting of the $d(d-1)/2$ unique off-diagonal entries of $\mathbf{R}$. Moreover, we denote with $\mathcal{K}_d\subset\mathbb{R}^{d\times d}$ the set of $d\times d$ correlation matrices.\footnote{
%
%In the following, we will thus write $\boldsymbol{\varrho}\in\mathcal{K}_d$ to indicate that $\mathbf{R}\in\mathcal{K}_d$.
%\footnote{
Throughout this section, we will use $\boldsymbol{\varrho}$ as the parameter in likelihood-based inference, and we will write $\mathbf{R}$ as shorthand for $\mathbf{R}(\boldsymbol{\varrho})$ when working with matrix operations. This slight abuse of notation simplifies the exposition and is justified by the one-to-one correspondence $\boldsymbol{\varrho} \leftrightarrow \mathbf{R}$.
}
%We denote with $\mathcal{K}_d\subset\mathbb{R}^{d\times d}$ the set of $d\times d$ correlation matrices. %and
%
%
%In the following, we will thus 
%write $\boldsymbol{\varrho}\in\mathcal{K}_d$ to indicate that $\mathbf{R}\in\mathcal{K}_d$.

\smallskip
The Gaussian copula function and its density are defined in terms of the standard multivariate normal distribution. For a vector $\boldsymbol{u} = [u_1, \dots, u_d]^\top \in [0,1]^d$, the copula function is given by:
\begin{equation}
    C( u_1, \dots,  u_d; \mathbf{R} )
    =
    \Phi_d( \Phi^{-1}(u_1), \dots, \Phi^{-1}(u_d); \mathbf{R} ) \, ,
\end{equation}
where $\Phi_d(\bigcdot; \mathbf{R})$ is the cdf of the centered $d$-variate Gaussian distribution with covariance matrix $\mathbf{R}$ \citep[see, e.g.,][Sec.4.8, p.147]{CLV}.
The corresponding copula density is given by:
\begin{equation}
\label{eq:NormalCopulaDensity}
    c ( u_1, \dots, u_d; \mathbf{R}) 
    %= 
    %\frac{
    %    \varphi_d( \Phi^{-1}(u_1), \dots, \Phi^{-1}(u_d); \mathbf{R} )
    %}{
    %    \prod_{i=1}^{n}  \varphi( \Phi^{-1}(u_i) )
    %}
    =
    \frac{
        \varphi_d( \Phi^{-1}(u_1), \dots, \Phi^{-1}(u_d); \mathbf{R} )
    }{
        \varphi_d( \Phi^{-1}(u_1), \dots, \Phi^{-1}(u_d); \mathbf{I}_d )
    } \,,
\end{equation}
where $\varphi_d(\bigcdot\,; \mathbf{R})$ is the pdf of the centered $d$-variate Gaussian distribution with covariance matrix $\mathbf{R}$, and $\mathbf{I}_d$ is the $d \times d$ identity matrix.

To establish notation for subsequent results, we introduce the following conventions. For 
any index set $\mathcal{S} \subseteq \mathbb{S}$ with cardinality $\mathsf{s} := |\mathcal{S}|$, we define its complement as $\mathcal{T} := \mathcal{S}^{\,\complement} = \{1, \dots, d\} \setminus \mathcal{S}$.
Moreover, we define the vector $\mathbf{z} := \left[ \Phi^{-1}(u_1), \dots, \Phi^{-1}(u_d) \right]^\top$, and partition it into the sub-vectors:
\begin{equation}\label{eq:zs_and_zt}
    \mathbf{z}_{\mathcal{S}} := 
    [\Phi^{-1}(u_i)]_{i \in \mathcal{S}}
    \hspace{1cm}\textnormal{and}\hspace{1cm}
    \mathbf{z}_{\mathcal{T}} := 
    [\Phi^{-1}(u_i)]_{i \in \mathcal{T}} ~.
\end{equation}
The correlation matrix $\mathbf{R}$ is partitioned conformably as:
\begin{equation}\label{eq:R_partitioning}
    \mathbf{R} =
    \begin{bmatrix}
        \mathbf{R}_{\mathcal{S} \mathcal{S}} &
        \mathbf{R}_{\mathcal{S} \mathcal{T}} \\
        \mathbf{R}_{\mathcal{T} \mathcal{S}} &
        \mathbf{R}_{\mathcal{T} \mathcal{T}}
    \end{bmatrix},
\end{equation}
where, for example, $\mathbf{R}_{\mathcal{S} \mathcal{S}}$ is the $\mathsf{s} \times \mathsf{s}$ correlation sub-matrix for the components in $\mathcal{S}$.\,\footnote{
It is well known from linear algebra that, if $\mathbf{R}$ is positive definite, then its principal submatrices $\mathbf{R}_{\mathcal{S}\mathcal{S}}$ and $\mathbf{R}_{\mathcal{T}\mathcal{T}}$
are also positive definite, and hence invertible.
}

\medskip
%\subsection{Closed-form likelihood expression}
The following Lemma provides a closed form-expression of the mixed partial derivative of the Gaussian copula. This result is fundamental for computing the model log-likelihood, as it provides the key copula term appearing in \eqref{eq:log_likelihood_CB}, leading to a closed-form expression for the log-likelihood.

\begin{lemma} \label{lemma:GaussianCopulaDerivative}
Let $C(u_1, \dots, u_d; \mathbf{R})$ be a $d$-dimensional Gaussian copula. The mixed partial derivative of $C$ w.r.t.\ the arguments in the index set $\mathcal{S}$, as it appears in \eqref{eq:log_likelihood_CB}, is given by
\begin{equation} \label{eq:gaussian_mixed_drvtv}
    \frac{\partial^{\, \mathsf{s}} C(u_1, \dots, u_d; \mathbf{R})}
     {\prod_{i \in \mathcal{S}} \partial u_i}
    =
    \frac{
    \varphi_{\mathsf{s}}\left( \mathbf{z}_{\mathcal{S}}; \mathbf{R}_{\mathcal{S}\mathcal{S}} \right)
    }{
    \varphi_{\mathsf{s}}\left( \mathbf{z}_{\mathcal{S}}; \mathbf{I}_{\mathsf{s}} \right)
    }
    \,\,
    \Phi_{d-\mathsf{s}}\left(
    \mathbf{z}_{\mathcal{T}} -
    \mathbf{R}_{\mathcal{T}\mathcal{S}} \mathbf{R}_{\mathcal{S}\mathcal{S}}^{-1} \mathbf{z}_{\mathcal{S}};\,
    \mathbf{R}_{\mathcal{T}\mathcal{T}} - \mathbf{R}_{\mathcal{T}\mathcal{S}} \mathbf{R}_{\mathcal{S}\mathcal{S}}^{-1} \mathbf{R}_{\mathcal{S}\mathcal{T}}
    \right).
\end{equation}
\end{lemma}
\begin{proof}
    See Appendix \ref{sec:Proofs}
\end{proof}

\noindent
Notably, the right-hand-side of \eqref{eq:gaussian_mixed_drvtv} can be interpreted as the product of two distinct factors:
the first term, $\displaystyle {
    \varphi_{\mathsf{s}}\left( \mathbf{z}_{\mathcal{S}}; \mathbf{R}_{\mathcal{S}\mathcal{S}} \right)
} \,/\,{
    \varphi_{\mathsf{s}}\left( \mathbf{z}_{\mathcal{S}}; \mathbf{I}_{\mathsf{s}} \right)
}$, is the copula density restricted only to the components in $\mathcal{S}$; the second term is the conditional cdf of the sub-vector $\mathbf{Z}_{\mathcal{T}}$ given that $\mathbf{Z}_{\mathcal{S}} = \mathbf{z}_{\mathcal{S}}$, i.e.\ %Specifically:
    \[
    \mathbb{P}\left( \mathbf{Z}_{\mathcal{T}} \leq \mathbf{z}_{\mathcal{T}} \mid \mathbf{Z}_{\mathcal{S}} = \mathbf{z}_{\mathcal{S}} \right) = \Phi_{d-\mathsf{s}}\left(
    \mathbf{z}_{\mathcal{T}} -
    \mathbf{R}_{\mathcal{T}\mathcal{S}} \mathbf{R}_{\mathcal{S}\mathcal{S}}^{-1} \mathbf{z}_{\mathcal{S}};\,
    \mathbf{R}_{\mathcal{T}\mathcal{T}} - \mathbf{R}_{\mathcal{T}\mathcal{S}} \mathbf{R}_{\mathcal{S}\mathcal{S}}^{-1} \mathbf{R}_{\mathcal{S}\mathcal{T}}
    \right).
    \]

\iffalse
Notably, the right-hand-side of \eqref{eq:gaussian_mixed_drvtv} can be interpreted as the product of two distinct factors:
\begin{itemize}
    \item The term $\displaystyle \frac{
    \varphi_{\mathsf{s}}\left( \mathbf{z}_{\mathcal{S}}; \mathbf{R}_{\mathcal{S}\mathcal{S}} \right)
}{
    \varphi_{\mathsf{s}}\left( \mathbf{z}_{\mathcal{S}}; \mathbf{I}_{\mathsf{s}} \right)
}$ is the copula density restricted only to the components in $\mathcal{S}$.
    \item The other term is the conditional cdf of the sub-vector $\mathbf{Z}_{\mathcal{T}}$ given that $\mathbf{Z}_{\mathcal{S}} = \mathbf{z}_{\mathcal{S}}$. Specifically:
    \[
    \mathbb{P}\left( \mathbf{Z}_{\mathcal{T}} \leq \mathbf{z}_{\mathcal{T}} \mid \mathbf{Z}_{\mathcal{S}} = \mathbf{z}_{\mathcal{S}} \right) = \Phi_{d-\mathsf{s}}\left(
    \mathbf{z}_{\mathcal{T}} -
    \mathbf{R}_{\mathcal{T}\mathcal{S}} \mathbf{R}_{\mathcal{S}\mathcal{S}}^{-1} \mathbf{z}_{\mathcal{S}};\,
    \mathbf{R}_{\mathcal{T}\mathcal{T}} - \mathbf{R}_{\mathcal{T}\mathcal{S}} \mathbf{R}_{\mathcal{S}\mathcal{S}}^{-1} \mathbf{R}_{\mathcal{S}\mathcal{T}}
    \right).
    \]
\end{itemize}
\fi

Lemma \ref{lemma:GaussianCopulaDerivative}, together with \eqref{eq:log_likelihood_CB}, provides an analytic expression for the log-likelihood in the Gaussian copula Comb-Bernoulli model.
The formula relies exclusively on the multivariate Gaussian cdf and pdf; hence, the numerical evaluation of the likelihood is highly efficient, as these functions are standard and optimized in all major mathematical software libraries.

\subsection{Asymptotic behavior and consistency with classical theory}

%A natural question arises: how does our framework relate to the classical copula setting when the probability masses at zero become negligible? This is particularly relevant for understanding the theoretical foundations of the Comb-Bernoulli approach and its connection to standard copula models with continuous marginals. We address this question by analyzing the limiting behavior of the log-likelihood as $\boldsymbol{p}$ tends to $\boldsymbol{1}$.

%The Comb-Bernoulli framework extends classical copula theory to accommodate probability masses at zero. This naturally raises the question: how does the model behave when these masses become negligible?

Having established the closed-form likelihood of the Gaussian Comb-Bernoulli model, we now turn to a fundamental theoretical question: how does the model behave when the probability masses at zero become negligible? This aspect is crucial for understanding the connection between the Comb-Bernoulli framework and standard copula models with continuous marginals. We address this question by analyzing the limiting behavior of the log-likelihood as $\boldsymbol{p}\to\boldsymbol{1}$.

\begin{proposition}\label{pr:conv_est_fcn_par}
For $\boldsymbol{p}\rightarrow\boldsymbol{1}$, the log-likelihood $\llcb(\boldsymbol{\Theta},\boldsymbol{\varrho}\mid\mathcal{X})$ in \eqref{eq:log_likelihood_CB} converges to the function
\begin{equation}
    \label{eq:limit_ll}    \llcbsp\left(\boldsymbol{\Theta},\boldsymbol{\varrho}\mid\mathcal{X}\right)
    :=
    \sum_{\boldsymbol{x}\in\mathcal{X}}
    \left[
    \ln 
    c\left(
    \Psi_1(x_1;\boldsymbol{\theta}_1),
    \dots,
    \Psi_d(x_d;\boldsymbol{\theta}_d);
    \boldsymbol{\varrho}
    \right)
    +
    \sum_{i\in\mathcal{S}(\boldsymbol{x})}
    \ln\big(
    \psi_i(x_i;\boldsymbol{\theta}_i)\big)
    \right]\,\,,
\end{equation}
almost surely for any $\boldsymbol{\varrho}\in\mathcal{K}_d$, where $c(\bigcdot;\boldsymbol{\varrho})$ is the Gaussian copula density in \eqref{eq:NormalCopulaDensity}.
\end{proposition}
\begin{proof}
	See Appendix \ref{sec:Proofs}
\end{proof}

The limiting log-likelihood $\llcbsp$ coincides with the standard log-likelihood for a Gaussian copula with continuous marginals. %thus establishing that our framework naturally encompasses the classical case. 
This convergence extends to the estimators themselves, as proved in the following theorem.

\begin{theorem}\label{prop:conv_estimators_tilde}
As $\boldsymbol{p}\to\boldsymbol{1}$, the IFM estimator $\rhocb$ \eqref{eq:argmax_p} of the Gaussian copula correlation converges almost surely to $\widehat{\boldsymbol{\varrho}}$, defined as %solution of %the following optimization problem 
\begin{equation}
    \label{eq:Asymtotic estimator}
    \widehat{\boldsymbol{\varrho}}:=\argmax_{\boldsymbol{\varrho}\in\mathcal{K}_d}\llcbsp\left(\widehat{\boldsymbol{\Theta}},\boldsymbol{\varrho}\mid\mathcal{X}\right)\,\,,
\end{equation}
where $\llcbsp$ is defined in \eqref{eq:limit_ll} and $\widehat{\boldsymbol{\Theta}}:=\left[\widehat{\boldsymbol{\theta}}_1,\ldots,\widehat{\boldsymbol{\theta}}_d\right]$ are the marginal MLEs.
\end{theorem}
\begin{proof}
	See Appendix \ref{sec:Proofs}
\end{proof}

Theorem \ref{prop:conv_estimators_tilde} is a major result of this paper: it establishes that the Comb-Bernoulli framework is a genuine generalization of the standard Gaussian copula with continuous marginals. When probability masses at zero vanish, our methodology recovers classical copula theory, ensuring consistency and compatibility with existing methods.

\medskip
Finally, the connection to the classical setting also allows us to benchmark the IFM estimator against another well-known dependence measure: Spearman's rank correlation. 
In the classical setting, Spearman's correlation provides an elegant and computational efficient route to estimating the entries of the copula matrix $\mathbf{R}$, as detailed in the following Proposition. 
%\citep[see, e.g.,][pp.502-505]{dodge2008concise}.

\begin{proposition}\label{cor:consis_with_clv}
Let $X_1,\ldots,X_d$ be random variables with continuous marginal distribution functions $\Psi_i(\cdot;\boldsymbol{\theta}_i)$ and suppose their dependence structure is governed by a Gaussian copula with correlation matrix $\mathbf{R} \in \mathcal{K}_d$. Then both the IFM estimator $\widehat{\boldsymbol{\varrho}}$ \eqref{eq:Asymtotic estimator} and the estimator
\[
\widehat{\mathrm{R}}_{ij} := 2\sin\left(\frac{\pi}{6}\widehat{\rho}_{ij}\right), \qquad i,j = 1,\ldots,d,
\]
where $\widehat{\rho}_{ij}$ denotes the sample Spearman correlation, are consistent for the entries of $\mathbf{R}$.
\end{proposition}

\begin{proof}
See Appendix \ref{sec:Proofs}.
\end{proof}

It is worth observing that this result does not extend to the Comb-Bernoulli framework:
the presence of probability masses at zero leads to the occurrence of tied observations, affecting the ranking used in Spearman's $\rho$ \citep[see, e.g.,][]{baviera2026spearman}. This renders it unsuitable when dealing with sparse financial time series and, consequently, for the estimation of the Gaussian Comb-Bernoulli parameters. 
The IFM procedure described in 
Section~\ref{ssec:estimation_of_the_copula_parameters} therefore emerges as the natural estimation route for the correlation in the Gaussian Comb-Bernoulli model;
its advantages over Spearman are demonstrated empirically in the next section.

%\medskip
%In summary, this Section has demonstrated that the Gaussian Comb-Bernoulli model possesses several important theoretical properties. First, the mixed partial derivatives admit closed-form expressions (Lemma \ref{lemma:GaussianCopulaDerivative}), enabling efficient computation of the log-likelihood. Second, as probability masses at zero vanish, both the log-likelihood and the resulting estimators converge to their classical counterparts (Proposition \ref{pr:conv_est_fcn_par} and Theorem \ref{prop:conv_estimators_tilde}). Third, the parametric IFM estimator maintains consistency with non-parametric rank-based alternatives in the continuous case (Corollary \ref{cor:consis_with_clv}). 

\section{Results}
\label{sec:Results}

In this section, we present an application of the Comb-Bernoulli model to one of the best-known insurance dataset, comparing its performance against standard benchmarks. After describing the dataset employed in our analysis, we discuss the results of parameter estimation, and finally assess the computational efficiency of our approach.

%In this Section, we present the empirical results obtained by applying the Comb-Bernoulli model to a real-world insurance dataset. The primary goal of this analysis is to validate the performance and robustness of the proposed methodology through a rigorous comparison against %two 
%standard benchmarks within the literature. %: the zero-mixed model and Spearman's rank correlation.

%\smallskip
%The Section is organized as follows. Section \ref{sec:dataset} details the properties and descriptive statistics of the dataset employed in our analysis. Section \ref{sec:estimation} presents the principal results on the estimation of the model parameters. Finally, Section \ref{sec:simulation} assesses the computational efficiency of the framework, providing a direct comparison between the simulation performance of the Comb-Bernoulli and L\'evy copula models.

\subsection{Dataset}\label{sec:dataset}
We consider the well-known Danish fire insurance dataset, which comprises claims (in Danish Krone, DKK) recorded between $1980$ and $1990$, gathered from the Copenhagen Reinsurance Company. %, resulting from fire insurance claims each exceeding 1 million Danish Krone (DKK)
%\textcolor{blue}{[Are we sure? Why is min smaller?]}. 
These claims are provided on a daily basis (a total of $4018$ days) and divided into three categories: ``Building", ``Contents", and ``Profit" claims. In Table \ref{tab:descr_stat}, we report descriptive statistics for the dataset; we refer the reader interested in a more detailed description of the dataset to the textbook of \citet{McNeil2005}.

%\begin{table}[H]
%\centering
%\begin{tabular}{|c|c|c|c|c|c|}
%\hline
% & \textbf{Observations} & \textbf{Average} & \textbf{Min} & \textbf{Max} & \textbf{Non-zero entries} \\
%\hline
%\textbf{Building} & $1541$ & $0.33$ Mln\euro  & $7.91\,K$\euro & $25.43$ Mln\euro & $38\%$ \\
%\textbf{Contents} & $1363$ & $0.27$ Mln\euro & $1.36\,K$\euro & $17.16$ Mln\euro & $34\%$ \\
%\textbf{Profits} & $561$ & $0.12$ Mln\euro & $0.53\,K$\euro & $8.05$ Mln\euro & $14\%$ \\
%\hline
%\end{tabular}
%\caption{Descriptive statistics for the three operational risk categories analysed.}
%\label{tab:descr_stat}
%\end{table}

\begin{table}[H]
\centering
\begin{tabular}{ccccc}
\toprule
 & \textbf{Observations} & \textbf{Mean} & \textbf{Min} & \textbf{Max} \\
\midrule
\textbf{Building} & $1541$ ($38\%$) & $2.56$ Mln DKK  & $60.85\,K$ DKK & $158.93$ Mln DKK \\
\textbf{Contents} & $1363$ ($34\%$) & $2.10$ Mln DKK & $10.47\,K$ DKK & $132.01$ Mln DKK \\
\textbf{Profits} & $561$ ($14\%$) & $0.94$ Mln DKK & $4.08\,K$ DKK & $61.93$ Mln DKK \\
\bottomrule
\end{tabular}
\caption{Descriptive statistics for the three risk categories: number of non-zero observations (with percentage out of 4018 days), along with mean, minimum, and maximum claim sizes.}
\label{tab:descr_stat}
\end{table}

Table \ref{tab:jumps_subtables} reports, for each bivariate time series, the number of no-jumps (both series equal to zero) and co-jumps (both series strictly positive). Table \ref{tab:jumps_tri} shows the corresponding counts for the trivariate series. 
%We observe that the number of no-jumps in the trivariate time series equals the number of no-jumps in the bivariate series ``Building"-``Contents". This means that, when these two series observe at the same time a zero-claim, the series ``Profits" also has a zero-claim. 
Interestingly, the number of no-jumps in the trivariate series matches the number of no-jumps in the bivariate series ``Building''–``Contents''. This indicates that whenever both the ``Building'' and ``Contents'' series record a zero claim simultaneously, the ``Profits'' series also records a zero claim.

\begin{table}[H]
\centering
%\begin{subtable}[t]{0.48\textwidth}
\centering
\begin{tabular}{cc c cc c c}
\toprule
 & \multicolumn{3}{c}{$\boldsymbol{\#}$ \textbf{No-jumps}} & \multicolumn{3}{c}{$\boldsymbol{\#}$ \textbf{Co-jumps}} \\
 \cmidrule(lr){2-4}\cmidrule(lr){5-7}
 & \textbf{Building} & \textbf{Contents} & \textbf{Profits} & \textbf{Building} & \textbf{Contents} & \textbf{Profits} \\ 
\midrule
\textbf{Building} &  & $2373$ & $2435$ & & $1259$ & $519$ \\
\textbf{Contents} &  &  & $2650$ & & & $556$ \\
\textbf{Profits}  &  &  & &  &  & \\
\bottomrule
\end{tabular}
%\caption{No-jumps}
%\end{subtable}
\iffalse
\hfill
\begin{subtable}[t]{0.48\textwidth}
\centering
\begin{tabular}{|c|c c c|}
\hline
 & \textbf{Building} & \textbf{Contents} & \textbf{Profits} \\ 
\hline
\textbf{Building} & & $1259$ & $519$ \\
\textbf{Contents} & & & $556$ \\
\textbf{Profits}  & & & \\
\hline
\end{tabular}
\caption{Co-jumps}
\end{subtable}
\fi
\caption{Number of no-jumps and co-jumps in each bivariate time series.}
\label{tab:jumps_subtables}
\end{table}

%
%In Tables \ref{tab:jumps_subtables}a and \ref{tab:jumps_subtables}b, we report, for each pair of risk categories, the number of days on which both series recorded zero claims (no-jumps) and the number of days on which both series recorded positive claims (co-jumps), respectively.

%
\iffalse
%
\begin{table}[H]
\centering
\begin{tabular}{|c | c c c|}
\hline
 & \textbf{Building} & \textbf{Contents} & \textbf{Profits} \\ 
\hline
\textbf{Building} & & $2373$ & $2435$ \\
\textbf{Contents} & & & $2650$ \\
\textbf{Profits}  & & & \\
\hline
\end{tabular}
\caption{Number of no-jumps in each bivariate time series.}
\label{tab:no_jumps}
\end{table}

\begin{table}[H]
\centering
\begin{tabular}{|c | c c c|}
\hline
 & \textbf{Building} & \textbf{Contents} & \textbf{Profits} \\ 
\hline
\textbf{Building} & & $1259$ & $519$ \\
\textbf{Contents} & & & $556$ \\
\textbf{Profits}  & & & \\
\hline
\end{tabular}
\caption{Number of co-jumps in each bivariate time series.}
\label{tab:co_jumps}
\end{table}
%
\fi
%
%: this fact is reasonable because to observe a claim due to lost profits, some damages either to the building or the content should have been observed.

\begin{table}[H]
    \centering
    \begin{tabular}{cc}
    \toprule
        $\boldsymbol{\#}$ \textbf{No-jumps} & $\boldsymbol{\#}$ \textbf{Co-jumps} \\
     \midrule
        $2373$ & $514$ \\
    \bottomrule
    \end{tabular}
    \caption{
    Number of no-jumps and co-jumps in the trivariate time series.
    } 
    %The number of no-jumps in the trivariate time series equals the number of no-jumps in the bivariate series ``Building"-``Contents".
    \label{tab:jumps_tri}
\end{table}

\subsection{Parameters' estimation} \label{sec:estimation}

We consider a trivariate Comb-Bernoulli model with lognormal marginals \eqref{eq:lognormal_marginals} and Gaussian copula (see Section~\ref{sec:gaussian_copula}).
Below, we present the estimated parameters and their confidence intervals obtained via parametric bootstrap.
Model calibration is performed using the IFM approach in Section \ref{sec:model}: we first fit the three marginal distributions, then the copula. 
For the explicit log-likelihood formulas, we refer the reader to Appendix~\ref{sec:triv_case}. 

%We then compare these results against two benchmarks from the literature: the zero-mixed model and the Gaussian copula correlation estimator derived from Spearman's rank correlation (see Corollary~\ref{cor:consis_with_clv}), discussing the advantages of our proposed method.

%In this Section, we present the estimated parameters of the Comb-Bernoulli model and their CIs obtained via a parametric bootstrap --~for the explicit log-likelihood formulas used in the trivariate case, we refer the reader to Appendix \ref{sec:triv_case}. Moreover, we compare these results against two benchmarks introduced previously: the zero-mixed model and the Gaussian copula correlation estimator derived from Spearman's rank correlation (see Corollary \ref{cor:consis_with_clv}).
%In this Section, we show the estimated values of the models' parameters. We show first the marginal parameters of the Comb-Bernoulli model, then we show the calibrated values of the Gaussian copula correlation
%in both models, and in both the bivariate and trivariate cases.
% \textcolor{blue}{Rimando ad appendix B per le formule, ricordare con cosa confrontiamo.}

%\smallskip
Table \ref{tab:marginal_params} collects the estimated marginal parameters: the jump probabilities $\widehat{\boldsymbol{p}}$, 
%Poisson intensities $\widehat{\boldsymbol{\lambda}}$ 
and the two parameters of the lognormal distribution, the vector of means $\widehat{\boldsymbol{\mu}}$ and the standard deviations $\widehat{\boldsymbol{\sigma}}$, for the three marginal distributions.

\begin{table}[H]
\centering
\begin{tabular}{lccc}
\toprule
 & \multicolumn{3}{c}{\textbf{Parameter estimates}} \\
\cmidrule(lr){2-4}
\textbf{Category} & $\widehat{\boldsymbol{p}}$ & \hspace{0.4cm}$\widehat{\boldsymbol{\mu}}$ & $\widehat{\boldsymbol{\sigma}}$ \\
\midrule
\textbf{Building} & $0.38$ & \hspace{0.4cm}$0.55$ & $0.81$ \\
\textbf{Contents} & $0.34$ & \hspace{0.4cm}$-0.21$ & $1.30$ \\
\textbf{Profits} & $0.14$ & \hspace{0.4cm}$-1.19$ & $1.43$ \\
\bottomrule
\end{tabular}
\caption{Estimated parameters of the Comb-Bernoulli marginal models: jump probabilities ($\widehat{\boldsymbol{p}}$), and lognormal means ($\widehat{\boldsymbol{\mu}}$) and standard deviations ($\widehat{\boldsymbol{\sigma}}$).}
\label{tab:marginal_params}
\end{table}

\smallskip
Once fitted the marginals, we estimate the correlation matrix of the Gaussian copula for the complete trivariate time series. Table \ref{tab:corr_cb_td} shows the point estimates of the dependence parameters along with their $95\%$ CIs. These intervals are obtained via a parametric bootstrap with $B=10^3$ replicas, using Algorithm \ref{al:parametric}.
The estimated correlation values indicate a strong positive dependence across all risk pairs, ranging from $0.667$ to $0.789$. Notably, the $95\%$ CIs for the correlation parameters are relatively narrow: these tight bounds suggest that the correlation estimates provided by the Comb-Bernoulli model are statistically precise and highly reliable, despite the high sparsity of the dataset.
\begin{table}[H]
\centering
\begin{tabular}{ccccccc}
\toprule
 & \multicolumn{3}{c}{\textbf{Correlation estimates}} &
 \multicolumn{3}{c}{\textbf{Confidence Intervals}}\\
\cmidrule(lr){2-4}\cmidrule(lr){5-7}

 & \textbf{Building} & \textbf{Contents} & \textbf{Profits} &\textbf{Building} & \textbf{Contents} & \textbf{Profits} \\
\midrule
\textbf{Building}  & & $0.761$ & $0.667$ &  & $[0.739;\,0.783]$ & $[0.629;\,0.699]$ \\
\textbf{Contents}  &  & & $0.789$ & & & $[0.757;\,0.813]$ \\
\textbf{Profits} & & & & & & \\
\bottomrule
\end{tabular}
\iffalse
\begin{tabular}{|c | c c c||}
\hline
 & \textbf{Building} & \textbf{Contents} & \textbf{Profits} \\ 
\hline
\textbf{Building} & & $0.761$ & $0.667$ \\
\textbf{Contents} & & & $0.789$ \\
\textbf{Profits}  & & & \\
\hline
\end{tabular}
\hspace{-0.3cm}
\begin{tabular}{| c c c|}
\hline
  \textbf{Building} & \textbf{Contents} & \textbf{Profits} \\ 
\hline
 & $[0.739;0.783]$ & $[0.629;0.699]$ \\
 & & $[0.757;0.813]$ \\
 & & \\
\hline
\end{tabular}
\fi
\caption{Gaussian copula correlation estimates ($\hat{\boldsymbol{\varrho}}_{\boldsymbol{p}}$) and corresponding $95\%$ confidence intervals for the Comb-Bernoulli model. The confidence intervals are calculated via a parametric bootstrap with $B=10^3$ replicas.}
\label{tab:corr_cb_td}
\end{table}

\subsubsection{Comparison with benchmarks}
To evaluate the effectiveness of the Comb-Bernoulli model, we compare our results against two key benchmarks: the zero-mixed model and Spearman's correlation. We focus in particular on the estimate of the dependence structure

\smallskip

First, we calibrate the dependence structure using the zero-mixed model. The calibration results, presented in Table \ref{tab:zero_mixed_full}, explicitly highlight the severe limitations of this approach in higher dimensions. From a parsimony perspective, the zero-mixed model is highly parameter-intensive; in a simple trivariate case ($d=3$), it requires estimating $14$ distinct parameters only to describe the occurrence and dependence structures. Furthermore, estimating these parameters accurately is difficult due to data sparsity. Because co-events in sparse time series are rare, calibration relies on very few data points, thereby severely destabilizing the estimation process. This sparsity also manifests in the precision of the estimates: the 95\% confidence intervals for the zero-mixed model are remarkably wide 
(see, e.g., $\hat{\varrho}_{13}^{(2)}$ in Table \ref{tab:zero_mixed_full}).

\begin{table}[t]
\centering
\begin{tabular}{clll}
\toprule
\textbf{Active Set} $\mathcal{S}$ & \textbf{Description} & \textbf{Prob. $\hat{p}_{\mathcal{S}}$ (95\% CI)} & \textbf{Correlation (95\% CI)} \\
\midrule
\multicolumn{4}{c}{\textit{Zero and Univariate Jumps (No dependence structure)}} \\
\midrule
$\varnothing$ & \textbf{No jumps} & $0.591$ $[0.568; 0.613]$ & \multicolumn{1}{c}{--} \\
$\{1\}$ & \textbf{Building} & $0.069$ $[0.057; 0.081]$ & \multicolumn{1}{c}{--} \\
$\{2\}$ & \textbf{Contents} & $0.015$ $[0.010; 0.021]$ & \multicolumn{1}{c}{--} \\
$\{3\}$ & \textbf{Profits} & $0.000$ $[0.000; 1.000]$ & \multicolumn{1}{c}{--} \\
\midrule
\multicolumn{4}{c}{\textit{Bivariate Co-jumps}} \\
\midrule
$\{1, 2\}$ & \textbf{Building \& Contents} & $0.185$ $[0.168; 0.203]$ & $\hat{\varrho}_{12}^{(2)} = $ $0.259$ $[0.157; 0.356]$ \\
$\{1, 3\}$ & \textbf{Building \& Profits}  & $0.001$ $[0.000; 0.003]$ & $\hat{\varrho}_{13}^{(2)} = $ $-0.376$ $[-0.985; 0.931]$ \\
$\{2, 3\}$ & \textbf{Contents \& Profits}  & $0.011$ $[0.006; 0.015]$ & $\hat{\varrho}_{23}^{(2)} = $ $0.599$ $[0.221; 0.820]$ \\
\midrule
\multicolumn{4}{c}{\textit{Trivariate Co-jumps}} \\
\midrule
\multirow{3}{*}{$\{1, 2, 3\}$} & \multirow{3}{*}{\textbf{All three classes}} & \multirow{3}{*}{$0.128$ $[0.113; 0.143]$} & $\hat{\varrho}_{12}^{(3)} = $ $0.329$ $[0.210; 0.439]$ \\
 & & & $\hat{\varrho}_{13}^{(3)} = $ $0.289$ $[0.167; 0.402]$ \\
 & & & $\hat{\varrho}_{23}^{(3)} = $ $0.639$ $[0.556; 0.709]$ \\
\bottomrule
\end{tabular}
\caption{Calibration results for the $3$-dimensional zero-mixed model. For each possible active set $\mathcal{S}$ (i.e., the set of classes presenting a claim), we report the estimated probability $\hat{p}_{\mathcal{S}}$. For sets with $|\mathcal{S}|\ge 2$, we also report the corresponding Gaussian copula correlation parameters, distinguished by superscript ($2$) for bivariate copulas and ($3$) for the trivariate copula, alongside their $95\%$ confidence intervals. This explicitly highlights the parameter-intensive nature of the zero-mixed approach, requiring $14$ parameters just for the occurrence and dependence structure in $d=3$, and demonstrates the high uncertainty in estimation—reflected in the exceptionally wide confidence intervals—caused by data sparsity.}
\label{tab:zero_mixed_full}
\end{table}% to check the consistency of the results in the continuous-time limit. 

\medskip
As a second benchmark, we consider the correlation estimator obtained from Spearman's rank correlation. As discussed in Section \ref{sec:gaussian_copula} (cf. Proposition \ref{cor:consis_with_clv}), when the marginals are strictly continuous, the Gaussian copula correlation can be estimated via the transformation $\widehat{\mathrm{R}}_{i,j} = 2 \sin(\frac{\pi}{6}\widehat{\rho}_{i,j})$, where $\widehat{\rho}_{i,j}$ is the sample Spearman's correlation. One might be tempted to use this transformed estimator because computing Spearman's correlation is straightforward and readily available in standard statistical libraries.

Exact minimum and maximum equivalent estimates of Spearman’s correlation can be obtained via elementary algorithms 
\citep[cf.][]{baviera2026spearman}. By applying the continuous monotonic transformation $2\sin(\frac{\pi}{6}\widehat{\rho}_{i,j})$ from Proposition \ref{cor:consis_with_clv} to these bounds, we can directly obtain the corresponding minimum and maximum equivalent estimations for the Gaussian copula correlation $\widehat{\mathrm{R}}_{i,j}$.

In Table \ref{tab:spearman}, we show these minimum-maximum equivalent estimations for the transformed Spearman’s correlations; we observe that in some cases the range is almost $1.5$, quite close to $2$, the maximum allowed range for a correlation. Let us underline that this fact is observed in one of the most studied datasets in the insurance sector.
\begin{table}[H]
\centering
\begin{tabular}{c c c c}
\toprule
 & \textbf{Building} & \textbf{Contents} & \textbf{Profits} \\ 
\midrule
\textbf{Building} & & $[0.169;0.942]$ & $[-0.496;0.968]$ \\
\textbf{Contents} & & & $[-0.411;0.985]$ \\
\textbf{Profits}  & & & \\
\bottomrule
\end{tabular}
\caption{Minimum and maximum equivalent estimations for the Gaussian copula correlation parameter $\widehat{\mathrm{R}}_{i,j}$. The bounds are derived by applying the monotonic transformation $\widehat{\mathrm{R}}_{i,j} = 2 \sin(\frac{\pi}{6}\widehat{\rho}_{i,j})$ to the exact lower and upper bounds of Spearman's rank correlation.}
\label{tab:spearman}
\end{table}

Furthermore, the observed narrowness in Comb-Bernoulli CIs is a remarkable result, especially if these CIs are compared to the interval between the minimum and maximum possible values of the Spearman-based correlation, presented in Table \ref{tab:spearman}.

\subsubsection{Additional comments on the bivariate case}
As a further comparison, we also estimate the three bivariate ($d=2$) time series separately to investigate whether the results align with those of the complete trivariate model ($d=3$). Table \ref{tab:corr_cb_bd} presents both the point estimates and the corresponding $95\%$ confidence intervals for the Gaussian copula correlations, calibrated for each pair of time series separately.

\begin{table}[H]
\centering
\begin{tabular}{ccccccc}
\toprule
 & \multicolumn{3}{c}{\textbf{Correlation estimates}} &
 \multicolumn{3}{c}{\textbf{Confidence Intervals}}\\
\cmidrule(lr){2-4}\cmidrule(lr){5-7}

 & \textbf{Building} & \textbf{Contents} & \textbf{Profits} &\textbf{Building} & \textbf{Contents} & \textbf{Profits} \\
\midrule
\textbf{Building}  & & $0.763$ & $0.634$ &  & $[0.740;\,0.786]$ & $[0.592;\,0.675]$ \\
\textbf{Contents}  &  & & $0.783$ & & & $[0.754;\,0.811]$ \\
\textbf{Profits} & & & & & & \\
\bottomrule
\end{tabular}
\iffalse
\begin{tabular}{|c | c c c|}
\hline
 & \textbf{Building} & \textbf{Contents} & \textbf{Profits} \\ 
\hline
\textbf{Building} & & $0.763$ & $0.634$ \\
\textbf{Contents} & & & $0.783$ \\
\textbf{Profits}  & & & \\
\hline
\end{tabular}
\fi
\caption{Gaussian copula correlation estimates ($\hat{\boldsymbol{\varrho}}_{\boldsymbol{p}}$) and corresponding $95\%$ confidence intervals for the bivariate Comb-Bernoulli models. The estimation is performed by calibrating the three bivariate time series separately. Both point estimates and confidence intervals are consistent with those obtained from the complete trivariate model.}
\label{tab:corr_cb_bd}
\end{table}
\iffalse
The correlation estimates and CIs obtained from these separate bivariate calibrations are remarkably consistent with those of the trivariate case (Table \ref{tab:corr_cb_td}); in fact, they appear almost indistinguishable. This equivalence demonstrates the empirical stability of the Comb-Bernoulli model to lower dimensions, maintaining the quality of the estimation without a significant loss of information.

%CIs
The CIs in the bivariate cases, considering each couple of time series at a time, are considered in Table \ref{tab:ci_cb_bd}.
\iffalse
\begin{table}[H]
\centering
\begin{tabular}{|c | c c c|}
\hline
 & \textbf{Building} & \textbf{Contents} & \textbf{Profits} \\ 
\hline
\textbf{Building} & & $[0.740;0.786]$ & $[0.592;0.675]$ \\
\textbf{Contents} & & & $[0.754;0.811]$ \\
\textbf{Profits}  & & & \\
\hline
\end{tabular}
\caption{$95\%$ CI for Gaussian copula correlation $\rhocb$ in the bivariate Comb-Bernoulli. Obtained from $B=10^3$ replicas of the dataset.}
\label{tab:ci_cb_bd}
\end{table}
\fi
%\subsection{Simulation and Confidence Intervals}
\fi
Both point estimations and CIs are statistically equivalent in the two considered cases: the complete model ($d=3$, see Table \ref{tab:corr_cb_td}) and $3$ separate two-dimensional models, considering a couple of time series at a time.  In extensive numerical experiments, we have observed that the estimation of the correlation parameters achieves similar results.

It can be useful to state a conjecture: estimating the complete Gaussian copula model or several low-dimensional models leads to similar results. This fact can have relevant practical consequences because the fit of a low-dimensional Gaussian copula can be fast and elementary in most cases.
This result is very relevant in practice and deserves a comment.

\subsection{Simulation}
\label{sec:simulation}

In practical applications, it is crucial to be able to simulate replicas of the original time series according to a given model, i.e.\ generating synthetic time series data with the same parameters as the original one according to a given model.
The parametric bootstrap presented in Section \ref{ssec:estimation_of_the_copula_parameters} is only an example of the broad use of simulation methods.

In Table \ref{tab:sim_times}, we compare the simulation times of the Comb-Bernoulli and the L\'evy copula models.\footnote{Computations were performed using Python 3.9 on a Windows 10 system with an Intel Core i5-9600K CPU (3.70~GHz) and 16~GB RAM.} Simulation time for the Comb-Bernoulli is dramatically faster than the L\'evy copula model, particularly as the dimensionality $d$ increases. As shown in Table \ref{tab:sim_times}, for $d=2$ the Comb-Bernoulli simulation time is similar to that of the L\'evy copula model, but it becomes more than twenty times faster when simulating trivariate time series ($d=3$). This result substantiates the claim that the Comb-Bernoulli process offers a straightforward simulation algorithm that overcomes the scalability limitations of the L\'evy copula approach, which requires simulating $2^d-1$ processes. The model achieves its goal of providing a parsimonious and computationally efficient tool for modelling dependent sparse time series.
\begin{table}[H]
    \centering
    \begin{tabular}{ccccc}
\toprule
  & $\boldsymbol{d=2}$ & $\boldsymbol{d=3}$ & $\boldsymbol{d=20}$ & $\boldsymbol{d=100}$ \\
\midrule
        \textbf{Comb-Bernoulli} & $2.27$ & $3.33$ & $20.37$ & $110.51$ \\
        \textbf{L\'evy copula} & $2.32$ & $76.30$ & - & - \\
\bottomrule
\end{tabular}
\iffalse
    \begin{tabular}{|c|c c c c|}
    \hline
         & $\boldsymbol{d=2}$ & $\boldsymbol{d=3}$ & $\boldsymbol{d=20}$ & $\boldsymbol{d=100}$ \\
    \hline
           \textbf{Comb-Bernoulli} & $2.27$ & $3.33$ & $20.37$ & $110.51$ \\
        \textbf{L\'evy copula} & $2.12$ & $76.30$ & - & - \\
    \hline
    \end{tabular}
    \fi
    \caption{Comparison of average simulation times (in seconds) for $B=10^3$ replicas of the Comb-Bernoulli and L\'evy copula models across different dimensions ($d=2$, $3$, $20$, $100$). The simulations are performed using Student-t copulas, with model parameters randomly drawn for each run. The reported times represent the average over $100$ independent experiments. Missing values for the L\'evy copula at $d=20$ and $d=100$ indicate that the simulation becomes computationally unfeasible due to the exponential scaling of the algorithm.}
    \label{tab:sim_times}
\end{table}

%\begin{table}[H]
%    \centering
%    \begin{tabular}{|c|}
%    \hline
%         $\boldsymbol{d=20}$  \\
%    \hline
%         $10.32$ \\
%          \\
%    \hline
%    \end{tabular}
%    \caption{Prova}
%    \label{tab:prova}
%\end{table}
This result is of extreme importance in the insurance sector. As an example, consider an internal model in Solvency II for capital requirements. Generally, the dependence between a large number of different risk classes should be taken into account; thus, a fast simulation scheme as that of the Comb-Bernoulli model can be very useful. For the proposed model, even for a large number of risk classes, the simulation time remains negligible. Indeed, considering $d=20$ and $d=100$ risk classes, % (roughly the number of Operational Risk factors in Solvency II), 
the average times for the simulation of $B=10^3$ replicas of the Comb-Bernoulli are respectively $20.37$ and $110.51$ seconds, which confirms the linear growth of computational time with the dimension $d$.

\section{Conclusions}
\label{sec:conclusions}

This paper introduces the \textit{Comb-Bernoulli model}, a novel multivariate framework designed to model sparse time series with mixed discrete-continuous marginals -- a common feature in insurance claims data where periods without claims (zeros) alternate with occasional positive claims. 
The model effectively bridges the gap between two established approaches: 
the flexible but parameter-intensive zero-mixed models, common in discrete-time settings, 
and the parsimonious but simulation-challenging L\'evy copula models used in continuous-time. 

The Comb-Bernoulli model is constructed by specifying each marginal as a mixture of a point mass at zero (with probability \(1-p_i\)) and a continuous severity distribution for positive claims (with probability \(p_i\)). The dependence structure across different risk classes is then introduced via a traditional copula function applied to these mixed marginal distributions. 

\smallskip

There are three main contributions of this paper. Firstly, we formally define the Comb-Bernoulli model and derive its closed-form likelihood function, enabling efficient parameter estimation via the IFM method. 
%\textcolor{red}{The model's two-step simulation algorithm scales linearly with dimensionality, requiring only $d$ draws from the chosen copula.}

Secondly, we provide a detailed analysis of the Gaussian copula specification. We prove that the IFM estimator of the correlation matrix converges almost surely to a well-defined limit as occurrence probabilities approach one. 
%and demonstrate that this limit coincides with the classical Spearman correlation estimator – establishing theoretical consistency with standard non-parametric methods.

Lastly, we validate the model empirically using the Danish fire insurance dataset. The results show not only that the model provides an accurate correlation estimate with tight confidence intervals, but also exceptional computational advantages. While the simulation algorithm for L\'evy copula models scales exponentially --~rendering simulations completely unfeasible for high dimensions such as $d=20$ or $d=100$~-- the Comb-Bernoulli model scales linearly. For instance, simulating $10^3$ replicas for a 20-dimensional portfolio (typical of Solvency II operational risk factors) takes only about $20$ seconds. This makes the Comb-Bernoulli model a highly practical, scalable, and computationally efficient tool for pricing and risk management in the insurance sector.

%\clearpage
%\newpage
\bibliography{sources}
\bibliographystyle{tandfx}

\section*{Acknowledgements}
We would like to thank all participants to the 1st ASTIN Bulletin Conference (Zurich), and
Workshop in Economics and Financial Mathematics (Verona), in particular Eric Andr\'e, Giuseppe
Buccheri, Corrado De Vecchi and Lorenzo Frattarolo for helpful suggestions.

\appendix

\bigskip\bigskip\bigskip
\noindent
{\LARGE\bfseries{Appendices}}

\renewcommand{\theequation}{\thesection.\arabic{equation}}
\renewcommand{\thetheorem}{\thesection.\arabic{theorem}}
\renewcommand{\thedefinition}{\thesection.\arabic{definition}}

\renewcommand{\thesection}{\Alph{section}}
\setcounter{section}{0} % reset section counter
\setcounter{lemma}{0}
\setcounter{equation}{0}

\vspace{0.6cm}

\noindent
The following appendices provide proofs and supporting material for the results presented in the paper. Appendices \ref{sec:biv_case} and \ref{sec:triv_case} illustrate, respectively, the formulas for the Comb-Bernoulli model in the bivariate and the trivariate case, while Appendix~\ref{sec:Proofs} contains the proofs.

\section{Comb-Bernoulli in the bivariate case}
\label{sec:biv_case}

The bivariate case ($d=2$) provides the simplest illustration of the Comb-Bernoulli framework. 
The following result, a direct corollary of Proposition \ref{pr:CombBernoulli_ll}, presents the specific form of the likelihood for this special case.

\begin{corollary} \label{prop:BivariateLikelihood}
In the bivariate case, the per-observation log-likelihood is given by:
\begin{equation} 
\label{eq:OriginalLikelihood}
\begin{split}
\ln\ell_{\boldsymbol{p}} (\boldsymbol{\Theta},\boldsymbol{\varrho}\mid x_1, x_2)
& =
\ln C(1-p_1, 1-p_2; \boldsymbol{\varrho}) 
    \mathbbm{1}_{\{x_1=0, \, x_2=0\}} ~~ +
\\[0.15cm]
& +
\ln \left[
    \frac{\partial C (u_1, 1-p_2; \boldsymbol{\varrho}) }{\partial u_1} 
    \right]_{u_1=F_1 \left(x_1;\,p_1\right)} ~ 
    \mathbbm{1}_{\{x_1>0, \, x_2=0\}} ~~ + \\
& +
\ln \left[
    \frac{\partial C (1-p_1, u_2; \boldsymbol{\varrho}) }{\partial u_2} 
    \right]_{u_2=F_2 \left(x_2;\,p_2\right)} ~
    \mathbbm{1}_{\{x_1=0, \, x_2>0\}} ~~ + \\[0.15cm]
& +
\ln c \left( F_1 (x_1;p_1),F_2 (x_2;p_2); \boldsymbol{\varrho} \right) \mathbbm{1}_{\{x_1>0,\, x_2>0\}} ~~ + \\[0.20cm]
& +
\ln \left[p_1 \psi_1 \left( x_1;\boldsymbol{\theta}_1 \right) \right] \mathbbm{1}_{\{x_1>0\}} +
\ln \left[p_2 \psi_2 \left( x_2;\boldsymbol{\theta}_2 \right) \right] \mathbbm{1}_{\{x_2>0\}} 
~,
\end{split}
\end{equation}
where
$F_1( \bigcdot;p_1 )$ and $F_2( \bigcdot;p_2 )$ 
are the unidimensional cdfs for the claims, as described in \eqref{eq:MarginalModel},   
and $\psi_1(\bigcdot;\boldsymbol{\boldsymbol{\theta}}_1)$ and $\psi_2(\bigcdot;\boldsymbol{\boldsymbol{\theta}}_2)$ are the marginal pdfs for the non-null claim sizes.
\end{corollary}

\iffalse
    The following identity holds in the 2-d case:
\[
    \overline{C}(u_1, u_2; \boldsymbol{\varrho} ) 
    + u_1 + u_2 - 
    C(u_1, u_2; \boldsymbol{\varrho} ) = 1
\]
Thus, \eqref{eq:OriginalLikelihood} can also be written as
%
\begin{equation}
\begin{split}
\ln\ell_{\boldsymbol{p}} (\boldsymbol{\Theta},\boldsymbol{\varrho}\mid x_1, x_2)
& =
\ln C(1-p_1, 1-p_2; \boldsymbol{\varrho}) 
    \mathbbm{1}_{\{x_1=0, \, x_2=0\}} ~~ +
\\[0.15cm]
& +
\ln \left[ 1 +
    \frac{\partial \overline{C} (u_1, 1-p_2; \boldsymbol{\varrho}) }{\partial u_1} 
    \right]_{u_1=F_1 \left(x_1;\,p_1\right)} ~ 
    \mathbbm{1}_{\{x_1>0, \, x_2=0\}} ~~ + \\
& +
\ln \left[ 1 + 
    \frac{\partial \overline{C} (1-p_1, u_2; \boldsymbol{\varrho}) }{\partial u_2} 
    \right]_{u_2=F_2 \left(x_2;\,p_2\right)} ~
    \mathbbm{1}_{\{x_1=0, \, x_2>0\}} ~~ + \\[0.15cm]
& +
\ln c \left( F_1 (x_1;p_1),F_2 (x_2;p_2); \boldsymbol{\varrho} \right) \mathbbm{1}_{\{x_1>0,\, x_2>0\}} ~~ + \\[0.20cm]
& +
\ln \left[p_1 \psi_1 \left( x_1;\boldsymbol{\theta}_1 \right) \right] \mathbbm{1}_{\{x_1>0\}} +
\ln \left[p_2 \psi_2 \left( x_2;\boldsymbol{\theta}_2 \right) \right] \mathbbm{1}_{\{x_2>0\}} 
~.
\end{split}
\end{equation}
%
\fi

\smallskip
In the Gaussian copula case, all the formulas above can be written explicitly. We recall that in the bidimensional Gaussian Comb-Bernoulli, the dependence is encoded in the correlation matrix
\[
\mathbf{R}:=\left[\begin{array}{cc}
    1 & \varrho \\
    \varrho & 1
\end{array}\right]\,\,.
\]
 The following Corollary provides the per-observation log-likelihood formula in the bidimensional Gaussian Comb-Bernoulli.
\begin{corollary}
    In the bivariate Gaussian Comb-Bernoulli, the per-observation log-likelihood is given by:
\begin{equation*} 
\label{eq:BidimGaussianLikelihood}
\begin{split}
\ln\ell_{\boldsymbol{p}} (\boldsymbol{\Theta},\boldsymbol{\varrho}\mid x_1, x_2)
& =
\ln \Phi_2(\Phi^{-1}(1-p_1), \Phi^{-1}(1-p_2); \mathbf{R}) 
    \mathbbm{1}_{\{x_1=0, \, x_2=0\}} ~~ +
\\[0.15cm]
& +
\ln \Phi\left(\frac{\Phi^{-1}(1-p_2)-\varrho\,\Phi^{-1}(F_1(x_1;p_1))}{\sqrt{1-\varrho^2}}\right) ~ 
    \mathbbm{1}_{\{x_1>0, \, x_2=0\}} ~~ + \\
& +
\ln \Phi\left(\frac{\Phi^{-1}(1-p_1)-\varrho\,\Phi^{-1}(F_2(x_2;p_2))}{\sqrt{1-\varrho^2}}\right) ~
    \mathbbm{1}_{\{x_1=0, \, x_2>0\}} ~~ - \\[0.15cm]
& 
-\frac{1}{2}\left(\ln(1-\varrho^2)+\frac{\varrho^2(\Phi^{-1}(F_1(x_1;p_1))^2+\Phi^{-1}(F_2(x_2;p_2))^2)}{(1-\varrho^2)} \right)\mathbbm{1}_{\{x_1>0,\, x_2>0\}} ~~ + \\[0.20cm]
&+\frac{\varrho \, \Phi^{-1}(F_1(x_1;p_1)) \, \Phi^{-1}(F_2(x_2;p_2))}{(1-\varrho^2)} \mathbbm{1}_{\{x_1>0,\, x_2>0\}}~~ + \\[0.15cm]
& +
\ln \left[p_1 \psi_1 \left( x_1;\boldsymbol{\theta}_1 \right) \right] \mathbbm{1}_{\{x_1>0\}} +
\ln \left[p_2 \psi_2 \left( x_2;\boldsymbol{\theta}_2 \right) \right] \mathbbm{1}_{\{x_2>0\}} 
~.
\end{split}
\end{equation*}
%\ln \varphi_2 \left( F_1 (x_1;p_1),F_2 (x_2;p_2); \mathbf{R} \right)
\end{corollary}

%\quad\\
\clearpage\newpage
\section{Comb-Bernoulli in the trivariate case}
\label{sec:triv_case}
In this Appendix, we illustrate the formulas in the three-dimensional case, noting that they do not become significantly more involved compared with the bivariate case. The following result, a direct corollary of Proposition \ref{pr:CombBernoulli_ll}, presents the specific form of the likelihood for this case.

\begin{corollary} \label{prop:TrivariateLikelihood}
In the trivariate case, the per-observation log-likelihood is given by:
\begin{equation} 
\label{eq:TridimOriginalLikelihood}
\begin{split}
\ln\ell_{\boldsymbol{p}} (\boldsymbol{\Theta},\boldsymbol{\varrho}&\mid x_1, x_2, x_3)
 =
\ln C(1-p_1, 1-p_2, 1-p_3; \boldsymbol{\varrho}) 
    \mathbbm{1}_{\{x_1=0, \, x_2=0, \, x_3=0\}} ~~ +
\\[0.15cm]
& +
\ln \left[
    \frac{\partial C (u_1, 1-p_2, 1-p_3; \boldsymbol{\varrho}) }{\partial u_1} 
    \right]_{u_1=F_1 \left(x_1;\,p_1\right)} ~ 
    \mathbbm{1}_{\{x_1>0, \, x_2=0,\,x_3=0\}} ~~ + \\
& +
\ln \left[
    \frac{\partial C (1-p_1, u_2,1-p_3; \boldsymbol{\varrho}) }{\partial u_2} 
    \right]_{u_2=F_2 \left(x_2;\,p_2\right)} ~
    \mathbbm{1}_{\{x_1=0, \, x_2>0,\,x_3=0\}} ~~ + \\[0.15cm]
& +
\ln \left[
    \frac{\partial C (1-p_1, 1-p_2,u_3; \boldsymbol{\varrho}) }{\partial u_3} 
    \right]_{u_3=F_3 \left(x_3;\,p_3\right)} ~
    \mathbbm{1}_{\{x_1=0, \, x_2=0,\,x_3>0\}} ~~ + \\[0.15cm]
& +
\ln \left[
    \frac{\partial^2 C (u_1, u_2,1-p_3; \boldsymbol{\varrho}) }{\partial u_1\partial u_2} 
    \right]_{u_1=F_1 \left(x_1;\,p_1\right),\,u_2=F_2 \left(x_2;\,p_2\right)} ~
    \mathbbm{1}_{\{x_1>0, \, x_2>0,\,x_3=0\}} ~~ + \\[0.15cm]
& +
\ln \left[
    \frac{\partial^2 C (u_1, 1-p_2,u_3; \boldsymbol{\varrho}) }{\partial u_1\partial u_3} 
    \right]_{u_1=F_1 \left(x_1;\,p_1\right),\,u_3=F_3 \left(x_3;\,p_3\right)} ~
    \mathbbm{1}_{\{x_1>0, \, x_2=0,\,x_3>0\}} ~~ + \\[0.15cm]
& +
\ln \left[
    \frac{\partial^2 C (1-p_1, u_2,u_3; \boldsymbol{\varrho}) }{\partial u_2\partial u_3} 
    \right]_{u_2=F_2 \left(x_2;\,p_2\right),\,u_3=F_3 \left(x_3;\,p_3\right)} ~
    \mathbbm{1}_{\{x_1=0, \, x_2>0,\,x_3>0\}} ~~ + \\[0.15cm]
& +
\ln c \left( F_1 (x_1;p_1),F_2 (x_2;p_2),F_3(x_3;p_3); \boldsymbol{\varrho} \right) \mathbbm{1}_{\{x_1>0,\, x_2>0,\,x_3>0\}} ~~ + \\[0.20cm]
& +
\sum_{i=1}^3\ln \left[p_i \psi_i \left( x_i;\boldsymbol{\theta}_i \right) \right] \mathbbm{1}_{\{x_i>0\}}
~,
\end{split}
\end{equation}
where
$F_1( \bigcdot;p_1 )$, $F_2( \bigcdot;p_2 )$ and $F_3(\bigcdot;p_3)$
are the unidimensional cdfs for the claims, as described in \eqref{eq:MarginalModel},   
and $\psi_1(\bigcdot;\boldsymbol{\theta}_1)$, $\psi_2(\bigcdot;\boldsymbol{\theta}_2)$ and $\psi_3(\bigcdot;\boldsymbol{\theta}_3)$ are the marginal pdfs for the non-null claim sizes.
\end{corollary}

As in the bidimensional Gaussian Comb-Bernoulli case, it is possible to write explicitly all the formulas above also for the tridimensional Gaussian Comb-Bernoulli. We recall that, in a three-dimensional Gaussian copula, the dependence is encoded in the correlation matrix
\[
\mathbf{R}:=\left[\begin{array}{ccc}
    1 & \varrho_{12} & \varrho_{13} \\
    \varrho_{12} & 1 & \varrho_{23} \\
    \varrho_{13} & \varrho_{23} & 1
\end{array}\right]\,\,.
\]
The following Corollary provides the per-observation log-likelihood formula in the tridimensional Gaussian Comb-Bernoulli.
\begin{corollary}
    In the trivariate Gaussian Comb-Bernoulli, the per-observation log-likelihood is:
    \begin{equation*} 
\label{eq:TridimGaussianLikelihood}
\begin{split}
\ln\,&\,\ell_{\boldsymbol{p}} (\boldsymbol{\Theta},\boldsymbol{\varrho}\mid x_1, x_2, x_3)
 =
\ln \Phi_3(\Phi^{-1}(1-p_1), \Phi^{-1}(1-p_2), \Phi^{-1}(1-p_3); \mathbf{R}) 
    \mathbbm{1}_{\{x_1=0, \, x_2=0, \, x_3=0\}} ~~ +
\\[0.15cm]
& +
\ln \Phi_2\left(\left[\begin{array}{c}
    \Phi^{-1}(1-p_2)-\varrho_{12}\Phi^{-1}(F_1(x_1;p_1)) \\
    \Phi^{-1}(1-p_3)-\varrho_{13}\Phi^{-1}(F_1(x_1;p_1))
\end{array}\right];\left[\begin{array}{cc}
    1-\varrho_{12}^2 & \varrho_{23}-\varrho_{12}\varrho_{13} \\
    \varrho_{23}-\varrho_{12}\varrho_{13} & 1-\varrho_{13}^2
\end{array}\right]\right) ~ 
    \mathbbm{1}_{\{x_1>0, \, x_2=0,\,x_3=0\}} ~~ + \\
& +
\ln \Phi_2\left(\left[\begin{array}{c}
    \Phi^{-1}(1-p_1)-\varrho_{12}\Phi^{-1}(F_2(x_2;p_2)) \\
    \Phi^{-1}(1-p_3)-\varrho_{23}\Phi^{-1}(F_2(x_2;p_2))
\end{array}\right];\left[\begin{array}{cc}
    1-\varrho_{12}^2 & \varrho_{13}-\varrho_{12}\varrho_{23} \\
    \varrho_{13}-\varrho_{12}\varrho_{23} & 1-\varrho_{23}^2
\end{array}\right]\right) ~
    \mathbbm{1}_{\{x_1=0, \, x_2>0,\,x_3=0\}} ~~ + \\%[0.15cm]
& +
\ln \Phi_2\left(\left[\begin{array}{c}
    \Phi^{-1}(1-p_1)-\varrho_{13}\Phi^{-1}(F_3(x_3;p_3)) \\
    \Phi^{-1}(1-p_2)-\varrho_{23}\Phi^{-1}(F_3(x_3;p_3))
\end{array}\right];\left[\begin{array}{cc}
    1-\varrho_{13}^2 & \varrho_{12}-\varrho_{13}\varrho_{23} \\
    \varrho_{12}-\varrho_{13}\varrho_{23} & 1-\varrho_{23}^2
\end{array}\right]\right) ~
    \mathbbm{1}_{\{x_1=0, \, x_2=0,\,x_3>0\}} ~~ + \\%[0.15cm]
& +
\ln  c_2\left(\left[\begin{array}{c}
    \Phi^{-1}(F_1(x_1;p_1)) \\
    \Phi^{-1}(F_2(x_2;p_2))
\end{array}\right];\mathbf{R}_{\{1,2\},\{1,2\}}\right)~
    \mathbbm{1}_{\{x_1>0, \, x_2>0,\,x_3=0\}} ~~ + \\%[0.15cm]
& +
\ln \Phi\left(\frac{\Phi^{-1}(1-p_3)-\mathbf{R}_{\{3\},\{1,2\}}\mathbf{R}^{-1}_{\{1,2\},\{1,2\}}\left[\begin{array}{c}
    \Phi^{-1}(F_1(x_1;p_1)) \\
    \Phi^{-1}(F_2(x_2;p_2))
\end{array}\right]}{\sqrt{1-\mathbf{R}_{\{3\},\{1,2\}}\mathbf{R}^{-1}_{\{1,2\},\{1,2\}}\mathbf{R}_{\{1,2\},\{3\}}}}\right)~
    \mathbbm{1}_{\{x_1>0, \, x_2>0,\,x_3=0\}} ~~ + \\%[0.15cm]
& +
\ln c_2\left(\left[\begin{array}{c}
    \Phi^{-1}(F_1(x_1;p_1)) \\
    \Phi^{-1}(F_3(x_3;p_3))
\end{array}\right];\mathbf{R}_{\{1,3\},\{1,3\}}\right) ~
    \mathbbm{1}_{\{x_1>0, \, x_2=0,\,x_3>0\}} ~~ + \\%[0.15cm]
& +
\ln \Phi\left(\frac{\Phi^{-1}(1-p_2)-\mathbf{R}_{\{2\},\{1,3\}}\mathbf{R}^{-1}_{\{1,3\},\{1,3\}}\left[\begin{array}{c}
    \Phi^{-1}(F_1(x_1;p_1)) \\
    \Phi^{-1}(F_3(x_3;p_3))
\end{array}\right]}{\sqrt{1-\mathbf{R}_{\{2\},\{1,3\}}\mathbf{R}^{-1}_{\{1,3\},\{1,3\}}\mathbf{R}_{\{1,3\},\{2\}}}}\right) ~
    \mathbbm{1}_{\{x_1>0, \, x_2=0,\,x_3>0\}} ~~ + \\%[0.15cm]
& +
\ln c_2\left(\left[\begin{array}{c}
    \Phi^{-1}(F_2(x_2;p_2)) \\
    \Phi^{-1}(F_3(x_3;p_3))
\end{array}\right];\mathbf{R}_{\{2,3\},\{2,3\}}\right) ~
    \mathbbm{1}_{\{x_1=0, \, x_2>0,\,x_3>0\}} ~~ + \\%[0.15cm]
& +
\ln \Phi\left(\frac{\Phi^{-1}(1-p_1)-\mathbf{R}_{\{1\},\{2,3\}}\mathbf{R}^{-1}_{\{2,3\},\{2,3\}}\left[\begin{array}{c}
    \Phi^{-1}(F_2(x_2;p_2)) \\
    \Phi^{-1}(F_3(x_3;p_3))
\end{array}\right]}{\sqrt{1-\mathbf{R}_{\{1\},\{2,3\}}\mathbf{R}^{-1}_{\{2,3\},\{2,3\}}\mathbf{R}_{\{2,3\},\{1\}}}}\right) ~
    \mathbbm{1}_{\{x_1=0, \, x_2>0,\,x_3>0\}} ~~ + \\%[0.15cm]
& +
\ln c_3 \left( \Phi^{-1}(F_1 (x_1;p_1)),\Phi^{-1}(F_2 (x_2;p_2)),\Phi^{-1}(F_3(x_3;p_3)); \mathbf{R} \right) ~\mathbbm{1}_{\{x_1>0,\, x_2>0,\,x_3>0\}} ~~ + \\%[0.20cm]
& +
\sum_{i=1}^3\ln \left[p_i \psi_i \left( x_i;\boldsymbol{\theta}_i \right) \right] \mathbbm{1}_{\{x_i>0\}}
~,
\end{split}
\end{equation*}
where $c_d$ denotes the $d$-variate normal copula density, and $\mathbf{R}_{I,J}$ denotes to the sub-matrix of $\mathbf{R}$ formed by the rows indexed by $I$ and the columns indexed by $J$.\footnote{For example, $\mathbf{R}_{\{1,2\},\{1,2\}}$ is the square sub-matrix composed of the first two rows and columns of $\mathbf{R}$, while $\mathbf{R}_{\{3\},\{1,2\}}$ is the row vector comprising the elements in the first two columns of the third row.}
\end{corollary}

\clearpage\newpage

\clearpage\newpage
\section{Proofs}
\label{sec:Proofs}

\textbf{Proof of Proposition \ref{pr:CombBernoulli_ll}.}
\noindent
Let $\boldsymbol{x}$ be an observation of the Comb-Bernoulli random vector $\boldsymbol{X}$, and let
$\ellcp(\boldsymbol{\Theta},\boldsymbol{\varrho}\mid\boldsymbol{x})$ denote its contribution to the log-likelihood.
Due to the structure of the model, each component of the random vector $\boldsymbol{X}$ has a mixed discrete–continuous distribution. %with an atom at $0$ and an absolutely continuous density on $(0,\infty)$.

By the Lebesgue decomposition theorem, the joint law can be split into atomic and absolutely continuous parts along each component \citep[cf., e.g.,][Thm.~6.10, p.~121]{rudin1987real}. 
Consequently, for any observation
$\boldsymbol{x}\in[0,\infty)^d$ with active set $\mathcal{S}(\boldsymbol{x})$,
the likelihood is obtained by differentiating
the joint cdf w.r.t.\ the continuous coordinates
and evaluating the remaining coordinates at the atom $0$
\citep[cf., e.g.,][Thm.~A1.4, p.~456]{kallenberg2002foundations}. Thus, the log-likelihood reads
\[
    \ellcp(\boldsymbol{\Theta},\boldsymbol{\varrho}\mid\boldsymbol{x})
    =
    \ln\left(
    \, 
    \frac{\partial^{\,|\mathcal{S}(\boldsymbol{x})|}}{\prod_{i \in \mathcal{S}(\boldsymbol{x})} \partial x_i} \;
    \mathbb{P}\Big(
        \{X_i \le x_i \,\, 
        \forall i \in 
        \mathcal{S}(\boldsymbol{x}) \, \} \cap
        \{X_j = 0 \,\, 
        \forall j \in \mathcal{S}(\boldsymbol{x})^\complement\, \}
    \Big) \,
    \right) \,.
\]

\noindent

In the Comb-Bernoulli model, the joint cdf is given by the copula representation \eqref{eq:copula_definition}:
\[
\mathbb{P}(\boldsymbol{X} \le \boldsymbol{x}) = 
C\big(F_1(x_1;p_1), \dots, F_d(x_d;p_d); \boldsymbol{\varrho}\big),
\]
where $F_i(\,\bigcdot\,; p_i) = 1-p_i + p_i \Psi_i(\bigcdot)$, with $\Psi_i(\bigcdot)$ absolutely continuous with density $\psi_i(\bigcdot)$.
Hence, by applying the chain rule of derivatives, the per-observation log-likelihood becomes
\[
    \ellcp(\boldsymbol{\Theta},\boldsymbol{\varrho}\mid\boldsymbol{x})
    =
    \ln\left(\left[
    \frac{\partial^{\,|\mathcal{S}(\boldsymbol{x})|}}{\prod_{i \in \mathcal{S}(\boldsymbol{x})} \partial u_i} 
    C(u_1, \dots, u_d; \boldsymbol{\varrho})
    \right]_{u_i = F_i(x_i;p_i)}\right) 
    +
    \sum_{i \in \mathcal{S}(\boldsymbol{x})} \ln\left(p_i \, \psi_i(x_i;\boldsymbol{\theta}_i)\right)
    \,.
\]

%\noindent
%Moreover, if $\mathcal{S} = \varnothing$ (i.e.\ all coordinates at the discrete atom $0$), the likelihood reduces to
%\[
%    \ellcp(\boldsymbol{\Theta},\boldsymbol{\varrho}\mid\boldsymbol{x})
%    =
%    \ln\left(C(1-p_1, \dots, 1-p_d; \boldsymbol{\varrho})\right) \,.
%\]

\noindent
The log-likelihood $\mathcal{L}_{\boldsymbol{p}}\left( \boldsymbol{\Theta}, \boldsymbol{\varrho} \mid \mathcal{X}\right)$ is obtained by summing all contributions in the time series.
$\hfill\square$

\bigskip
\noindent
\textbf{Proof of Lemma \ref{lem:prob_active_set}.}
    For any active set $I \subseteq \mathbb{S}$, define
    \(A := \bigcap_{i \in I}\{X_i > 0\}\)
    and
    \(B_j := \{X_j > 0\}\)\,.
    Then, the event of interest reads
    \[
    E 
    := 
    \Big(\, \bigcap_{i \in I}\{X_i > 0\} \Big)
    \cap 
    \Big(\, \bigcap_{j \in I^{\complement}}\{X_j = 0\} \Big)
    = 
    A 
    \cap
    \bigcap_{j \in I^{\complement}} B_j^{\,\complement}
    = 
    A
    \cap 
    \Big(\, \bigcup_{j \in I^{\complement}} B_j\Big)^{\complement} 
    =
    A
    \setminus
    \Bigg(
    A
    \cap 
    \, \bigcup_{j \in I^{\complement}} B_j
    \Bigg)
    \,.
    \]
    %where the last equality follows from De Morgan's laws.
    %
    By the inclusion--exclusion principle \citep[see, e.g.,][p.~99]{Feller},
    \begin{align*}
    \mathbb{P} \Big( A \cap \bigcup_{j \in I^{\complement}} B_j \Big) &=
    \mathbb{P} \Big( \bigcup_{j \in I^{\complement}} 
    \big( A \cap B_j \big) \Big) 
    =
    \\
    &=
    \sum_{\begin{subarray}{l}
        J \subseteq I^{\complement} \\[0.05cm]
        J \neq \varnothing
    \end{subarray}}
    (-1)^{|J|+1} \,\,
    \mathbb{P}\left( \bigcap_{j \in J} (A \cap B_j)\right)
    =
    \sum_{\begin{subarray}{l}
        J \subseteq I^{\complement} \\[0.05cm]
        J \neq \varnothing
    \end{subarray}}
    (-1)^{|J|+1} \, \mathbb{P} \Big( A \cap \bigcap_{j \in J} B_j \Big) \,.
    \end{align*}

    \noindent
    Therefore, using the identity $\mathbb{P}(E) = \mathbb{P}(A) - \mathbb{P}\big(A \cap \, \bigcup_{j \in I^{\complement}} \, B_j\big)$, we obtain
    \[
    \mathbb{P}(E) = 
    \sum_{ J \subseteq I^{\complement} }
    (-1)^{|J|} \, \mathbb{P} \Big( A \cap \bigcap_{j \in J} B_j \Big)
    =
    \sum_{ J \supseteq I }
    (-1)^{|J| - |I|} \, \mathbb{P} \Big( \bigcap_{j \in J} B_j \Big)
    \,.
    \]
    %where we adopt the convention that the intersection over the empty set is the entire sample space, so that the $J = \varnothing$ term yields $\mathbb{P}(A)$.

    \noindent
    Finally, under the copula representation of the joint distribution of $\mathbf{X}$, each probability satisfies
    \[
    \mathbb{P} \Big( \bigcap_{j \in J} B_j \Big)
    = 
    \overline{C}_{J} \Big( \{ \overline{F}_j(0) \}_{j \in J} \Big)
    =
    \overline{C}_{J} \Big( \{ p_j \}_{j \in J} \Big)
    \]
    where $\overline{C}_{J}$ denotes the restriction of the survival copula $\overline{C}$ to the coordinates in $J$.
$\hfill\square$

\bigskip
\noindent
\textbf{Proof of Proposition \ref{pr:CombPoisson_ll}.}
Given a time series $\mathcal{X}$, we rearrange the total likelihood, defined as
\[
\mathrm{L}_{\boldsymbol{p}}\left(\boldsymbol{\Theta},\boldsymbol{\varrho}\mid\mathcal{X}\right):=\exp\left\{\mathcal{L}_{\boldsymbol{p}}\left(\boldsymbol{\Theta},\boldsymbol{\varrho}\mid\mathcal{X}\right)\right\}\,\,,
\]
by grouping the observations according to their active set $\mathcal{S}(\boldsymbol{x})$. Taking the exponential of equation \eqref{eq:log_likelihood_CB} for the discrete likelihood $\mathrm{L}_{\boldsymbol{p}}$, we can rewrite it as:
\begin{equation}
\label{eq:ll_cb_rewritten}
    \mathrm{L}_{\boldsymbol{p}}(\boldsymbol{\Theta},\boldsymbol{\varrho}\mid\mathcal{X}) = (p_{\varnothing}^{\perp})^{n_{\varnothing}^{\perp}} \prod_{I \neq \varnothing} \left[ (p_I^{\perp})^{n_I^{\perp}} \left( \prod_{t=1}^N \psi_I^{\perp}(\boldsymbol{x}^{(t)}) \mathbb{I}_{\{\mathcal{S}(\boldsymbol{x}^{(t)}) = I\}} \right) \right]\,\,,
\end{equation}
where, thanks to the inclusion-exclusion principle,
\[
\psi_{I}^{\perp}( \boldsymbol{x})
    :=
    \frac{1}{p_I^{\perp} }
    \sum_{J \supseteq I}
    (-1)^{|J|}
    \;
    \frac{
        \partial^{\, |I|}
    }{
        \prod_{i \in I} \partial x_i
    }
    \overline{C}_J \left(\left\{ \overline{F}_i(x_i)\right\}_{i \in {J}} \right)
    \,\,.
\]
Since $\sum_{I \subseteq \mathbb{S}} p_I^{\perp} = 1$ and the total number of observations is $\sum_{I \subseteq \mathbb{S}} n_I^{\perp} = N$, the likelihood can be separated into three main factors:
\begin{equation}
\label{eq:ll_cb_rewritten2}
\mathrm{L}_{\boldsymbol{p}}(\boldsymbol{\Theta},\boldsymbol{\varrho}\mid\mathcal{X}) = \left(1 - \sum_{I \neq \varnothing} p_I^{\perp}\right)^N \left[ \prod_{I \neq \varnothing} \left( \frac{p_I^{\perp}}{1 - \sum_{I \neq \varnothing} p_I^{\perp}} \right)^{n_I^{\perp}} \right] \left[ \prod_{I \neq \varnothing} \left( \prod_{t=1}^N \psi_I^{\perp}\left(\boldsymbol{x}^{(t)}\right) \mathbb{I}_{\{\mathcal{S}(x^{(t)}) = I\}} \right) \right]\,\,.
\end{equation}

\bigskip
We set $p_I^{\perp} = \lambda_I^{\perp} \Delta t$ for every non-empty active set $I \neq \varnothing$, and we define $\lambda_i := \sum_{I \subseteq \mathbb{S}: i \in I} \lambda_I^{\perp}$.
We then consider the three terms in \eqref{eq:ll_cb_rewritten2} separately as $\Delta t \rightarrow 0$.
For the first term, it holds that
\[
\left(1 - \sum_{I \neq \varnothing} p_I^{\perp}\right)^N = \left(1 - \Delta t \sum_{I \neq \varnothing} \lambda_I^{\perp}\right)^{\frac{T}{\Delta t}} = \exp\left\{-T \sum_{I \neq \varnothing} \lambda_I^{\perp}\right\} + o(1)\,\,,
\]
while, for the second term, we have
\[
\prod_{I \neq \varnothing} \left( \frac{p_I^{\perp}}{1 - \sum_{I \neq \varnothing} p_I^{\perp}} \right)^{n_I^{\perp}} = \prod_{I \neq \varnothing} \left( \frac{\Delta t \lambda_I^{\perp}}{1 - \Delta t \sum_{I \neq \varnothing} \lambda_I^{\perp}} \right)^{n_I^{\perp}} = \left( \prod_{I \neq \varnothing} (\Delta t \lambda_I^{\perp})^{n_I^{\perp}} \right) (1+o(1))\,\,.
\]
For the last term, we consider the expression for $\psi_I^{\perp}(\boldsymbol{x})$, and knowing that the survival function can be written as $\overline{F}_i(x_i) = \Delta t\,\lambda_i\,\overline{\Psi}_i(x_i)$, we express the limit through the L\'evy copula $\mathfrak{C}_J$ as
\[
\psi_I^{\perp}(x) = \frac{1}{\lambda_I^{\perp}} \left( \prod_{i \in I} \lambda_i \psi_i(x_i) \right) \sum_{J \supseteq I} (-1)^{|I|+|J|} \left[ \frac{\partial^{|I|}}{\prod_{i \in I} \partial u_i} \mathfrak{C}_J\left(\{u_i\}_{i \in J}\right) \right]_{u_i = \lambda_i \overline{\Psi}_i(x_i)}
\]
Substituting these three limits back into equation \eqref{eq:ll_cb_rewritten2} yields the expansion
%all $\lambda_I^{\perp}$ at the denominator in $\psi_I^{\perp}$ simplify with the numerators from the second term. 
%The limit behavior of the likelihood is thus given by:
\begin{equation}
\label{eq:limit_probability_cb}
    %\mathrm{L}_{\Delta t\boldsymbol{\lambda}}(\boldsymbol{\Theta},\boldsymbol{\varrho}\mid\mathcal{X}) = 
    \prod_{I \neq \varnothing} \left[ e^{-T \lambda_I^{\perp}} \left( \prod_{i \in I} \lambda_i \right)^{n_I^{\perp}} \prod_{t=1}^{n_I^{\perp}} \left( \zeta_I\left(\boldsymbol{x}^{(t)}\right) \prod_{i \in I} \psi_i\left(x_i^{(t)}\right) \right) \right] (\Delta t)^{\sum_{I \neq \varnothing} n_I^{\perp}} + o(1)\,\,.
\end{equation}
When passing to continuous-time, 
%the probability of observing an exact sequence of jumps over an infinitely fine grid vanishes (represented by the isolated $(\Delta t)$ multiplier). Instead, 
the density of the process obtained is exactly the component within the square brackets in \eqref{eq:limit_probability_cb}, verifying the equation:
\[
\mathrm{L}_{\lambda}(\boldsymbol{\Theta},\boldsymbol{\varrho}\mid\mathcal{X}) 
= 
\prod_{I \neq \varnothing} 
\left[ \,
e^{-T \lambda_I^{\perp}} \left( \prod_{i \in I} \lambda_i \right)^{n_I^{\perp}} \, \prod_{t=1}^{n_I^{\perp}} \left( \zeta_I(\boldsymbol{x}^{(t)}) \prod_{i \in I} \psi_i\left(x_i^{(t)}\right) \right) 
\right]\,\,.
\]
Finally, by taking the logarithm, we get
the log-likelihood $\mathcal{L}_{\boldsymbol{\lambda}}$ in equation \eqref{eq:Comb_Poi_ll}.
$\hfill\square$

\bigskip
\noindent
\textbf{Proof of Lemma \ref{lemma:GaussianCopulaDerivative}.}
By the definition of the Gaussian copula, 
\[
    C(u_1, \dots, u_d; \mathbf{R}) = 
    \mathbb{P}(Z_1 \leq \Phi^{-1}(u_1), \dots, Z_d \leq \Phi^{-1}(u_d)) \,,
\]
where $\mathbf{Z} = [Z_1, \dots, Z_d]^\top$ is a $d$-dimensional st.n.\ random vector with correlation matrix $\mathbf{R}$.

\smallskip
For any $i=1,\dots,d$, let $z_i := \Phi^{-1}(u_i)$; the derivative of this transformation is
$\partial z_i / \partial u_i = 1 / \varphi(z_i)$.
Applying the chain rule for the partial derivative w.r.t.\ the variables $u_i$ for $i \in \mathcal{S}$ yields:
\begin{equation} \label{eq:partial_derivative_gaussian_proof}
    \frac{\partial^{\, \mathsf{s}} C(u_1, \dots, u_d; \mathbf{R})}
    {\prod_{i \in \mathcal{S}} \partial u_i}
    =
    \frac{1}{\prod_{i \in \mathcal{S}} \varphi(z_i)} \,
    \frac{\partial^{\, \mathsf{s}} \,
    \mathbb{P}(Z_1 \leq z_i, \dots, Z_d \leq z_d)}
    {\prod_{i \in \mathcal{S}} \partial z_i} \,\,.
\end{equation}

%\label{eq:R_partitioning}\label{eq:zs_and_zt}
\noindent
We partition the index set $\{1,\dots,d\}$ into the set $\mathcal{S}$ and its complement $\mathcal{T}$, and correspondingly partition the random vector and the correlation matrix as in 
\eqref{eq:zs_and_zt} and \eqref{eq:R_partitioning}:
\begin{equation*}
    \mathbf{Z} 
    =
    \begin{bmatrix}
        \mathbf{Z}_{\mathcal{S}} \\ \mathbf{Z}_{\mathcal{T}}
    \end{bmatrix}
    \quad
    \textnormal{and}
    \quad
    \mathbf{R} =
    \begin{bmatrix}
        \mathbf{R}_{\mathcal{S} \mathcal{S}} &
        \mathbf{R}_{\mathcal{S} \mathcal{T}} \\
        \mathbf{R}_{\mathcal{T} \mathcal{S}} &
        \mathbf{R}_{\mathcal{T} \mathcal{T}}
    \end{bmatrix} 
    \,.
\end{equation*}
By the fundamental theorem of calculus,
the joint derivative in \eqref{eq:partial_derivative_gaussian_proof} can be expressed as the joint density of $\mathbf{Z}_\mathcal{S}$ times the conditional probability of 
$\mathbf{Z}_\mathcal{T} = \mathbf{z}_\mathcal{T}$, i.e.,\ as
\[
    \frac{\partial \, 
        \mathbb{P}(\mathbf{Z}_\mathcal{S} \leq \mathbf{z}_\mathcal{S}, 
        \mathbf{Z}_\mathcal{T} \leq \mathbf{z}_\mathcal{T})}
    {\partial \mathbf{z}_\mathcal{S}}
    =
    \varphi_{\, \mathsf{s}}(\mathbf{z}_\mathcal{S}; 
    	\mathbf{R}_{\mathcal{S}\mathcal{S}}) \,
    \mathbb{P}(
    	\mathbf{Z}_\mathcal{T} \leq \mathbf{z}_\mathcal{T} 
    	\,\vert\, 
    	\mathbf{Z}_\mathcal{S} = \mathbf{z}_\mathcal{S}
    ) \,.
\]

\noindent
Moreover, the conditional distribution of $\mathbf{Z}_\mathcal{T}$ given
$\mathbf{Z}_\mathcal{S} = \mathbf{z}_\mathcal{S}$ satisfies the well-known formula
\[
    \mathbf{Z}_\mathcal{T} \,\vert\, 
    \mathbf{Z}_\mathcal{S} = 
    \mathbf{z}_\mathcal{S} \ 
    \sim\ 
    N_{d-\mathsf{s}}\big( 
        \mathbf{R}_{\mathcal{T}\mathcal{S}} 
        \mathbf{R}_{\mathcal{S}\mathcal{S}}^{-1} 
        \mathbf{z}_{\mathcal{S}}, 
        \mathbf{R}_{\mathcal{T}\mathcal{T}} - 
        \mathbf{R}_{\mathcal{T}\mathcal{S}} 
        \mathbf{R}_{\mathcal{S}\mathcal{S}}^{-1} 
        \mathbf{R}_{\mathcal{S}\mathcal{T}} \big) \,,
\]
and therefore
\begin{equation} \label{eq:conditional_prob_gaussian_proof}
    \mathbb{P}(
        \mathbf{Z}_\mathcal{T} \leq \mathbf{z}_\mathcal{T} 
    	\,\vert\,  
    	\mathbf{Z}_\mathcal{S} = \mathbf{z}_\mathcal{S}
    ) 
    = 
    \Phi_{d-\mathsf{s}}
    \big( 
    \mathbf{z}_\mathcal{T} - 
    \mathbf{R}_{\mathcal{T}\mathcal{S}} 
    \mathbf{R}_{\mathcal{S}\mathcal{S}}^{-1} 
    \mathbf{z}_{\mathcal{S}};\ 
    \mathbf{R}_{\mathcal{T}\mathcal{T}} - 
    \mathbf{R}_{\mathcal{T}\mathcal{S}} 
    \mathbf{R}_{\mathcal{S}\mathcal{S}}^{-1} 
    \mathbf{R}_{\mathcal{S}\mathcal{T}} 
    \big) \,.
\end{equation}

\noindent
By combining \eqref{eq:partial_derivative_gaussian_proof} and \eqref{eq:conditional_prob_gaussian_proof}, we obtain the stated expression. 
$\hfill\square$

\bigskip
\noindent
\textbf{Proof of Proposition \ref{pr:conv_est_fcn_par}}
Since, by Proposition \ref{pr:uniformly_convergence} below, the log-likelihood $\llcb\left(\boldsymbol{\Theta},\boldsymbol{\varrho}\mid\mathcal{X}\right)$ converges locally uniformly almost surely to $\llcbsp\left(\boldsymbol{\Theta},\boldsymbol{\varrho}\mid\mathcal{X}\right)$, then it also converges almost surely for all $\boldsymbol{\varrho}\in\mathcal{K}_d$.
$\hfill\square$

\smallskip
\noindent
\paragraph{Proof of Proposition \ref{prop:conv_estimators_tilde}} 
It is straightforward to prove that 
\[
\lim_{\boldsymbol{\varrho}\to\partial\mathcal{K}_d}\llcb\left(\widehat{\boldsymbol{\Theta}},\boldsymbol{\varrho}\mid\mathcal{X}\right)=-\infty\,\,,\quad\lim_{\boldsymbol{\varrho}\to\partial\mathcal{K}_d}\llcbsp\left(\widehat{\boldsymbol{\Theta}},\boldsymbol{\varrho}\mid\mathcal{X}\right)=-\infty\,\,,
\]
and that $\llcb(\widehat{\boldsymbol{\Theta}},\boldsymbol{\varrho}\mid\mathcal{X}),\,\llcbsp(\widehat{\boldsymbol{\Theta}},\boldsymbol{\varrho}\mid\mathcal{X})>-\infty$ $\forall\,\boldsymbol{\varrho}\in\mathcal{K}_d$, thus, $\forall\,M$, the sets
\[
K_{\boldsymbol{p},M}:=\left\{\boldsymbol{\varrho}\in\mathcal{K}_d\,|\,\llcb\left(\widehat{\boldsymbol{\Theta}},\boldsymbol{\varrho}\mid\mathcal{X}\right)\geq M\right\}\,\,,\quad K_M:=\left\{\boldsymbol{\varrho}\in\mathcal{K}_d\,|\,\llcbsp\left(\widehat{\boldsymbol{\Theta}},\boldsymbol{\varrho}\mid\mathcal{X}\right)\geq M\right\}
\]
are compact subsets of $\mathcal{K}_d$. 
Thanks to Proposition \ref{pr:uniformly_convergence} below, $\exists\,K\in\mathcal{K}_d,\,\overline{\boldsymbol{p}}$ such that $\rhocb\in K$ for $\boldsymbol{p}$ such that $p_i\geq\overline{p}_i$ for all $i=1,\ldots,d$.
Since, for $\boldsymbol{p}$ close to $\boldsymbol{1}$, the sequence $\{\rhocb\}_{\boldsymbol{p}}$ is contained in the compact set $K$, then every subsequence $\{\widehat{\boldsymbol{\varrho}}_{\boldsymbol{p}_m}\}_{m}$ converges to a $\tilde{\boldsymbol{\varrho}} \in K$. By definition of the $\arg\max$,
\[
\mathcal{L}_{\boldsymbol{p}_m}\left(\widehat{\boldsymbol{\Theta}},\widehat{\boldsymbol{\varrho}}_{\boldsymbol{p}_m}\mid\mathcal{X}\right) = \max_{\boldsymbol{\varrho} \in K} \mathcal{L}_{\boldsymbol{p}_m}\left(\widehat{\boldsymbol{\Theta}},\boldsymbol{\varrho}\mid\mathcal{X}\right) \ge \llcb\left(\widehat{\boldsymbol{\Theta}},\widehat{\boldsymbol{\varrho}}\mid\mathcal{X}\right)\,\,.
\]
Passing to the limit using locally uniform convergence on compacts, we obtain
\[
\llcbsp\left(\widehat{\boldsymbol{\Theta}},\tilde{\boldsymbol{\varrho}}\mid\mathcal{X}\right) \ge \llcbsp\left(\widehat{\boldsymbol{\Theta}},\widehat{\boldsymbol{\varrho}}\mid\mathcal{X}\right)\,\,,
\]
and uniqueness of the maximum implies $\tilde{\boldsymbol{\varrho}} = \widehat{\boldsymbol{\varrho}}$. Since every subsequence converges to the same $\widehat{\boldsymbol{\varrho}}$, the full sequence converges almost surely to $\widehat{\boldsymbol{\varrho}}$:
\[
    \phantom{\square}
    \hspace{6.5cm}
    \boldsymbol{\varrho}_p \overset{\text{a.s.}}{\longrightarrow} \widehat{\boldsymbol{\varrho}} \quad \text{as }\boldsymbol{p}\to\boldsymbol{1} \,.
    \hspace{6.5cm}
    \square
\]

\iffalse
\paragraph{Proof of Proposition \ref{prop:speraman_rho_copula}}
Let us define
\[
    u_{i} := \Psi_i \left(x_i,\boldsymbol{\theta}_i\right)\quad\forall\,i=1,...,d\,\,,
\]
and
\[
    \xi_{i} := \Phi^{-1} \left(\Psi_i \left(x_i,\boldsymbol{\theta}_i\right)\right)\quad\forall\,i=1,...,d\,\,,
\]
On the one hand, if the marginals are continuous, then the $u_i$ are uniform rvs while all the $ \xi_i $ are st.n.;
on the other hand, because the variables $[X_1,...,X_d]^\top$ are coupled with a Gaussian copula with correlation matrix $\mathbf{R}$, thanks to the Sklar's Theorem, 
\[
C(u_1,...,u_d; \mathbf{R}) = \Phi_d (\xi_{1},...,\xi_{d};\mathbf{R}) \,\,,
\]
that is equivalent to stating that
\[
    \left[ \begin{matrix} \xi_{1} \\\vdots\\ \xi_{d} \end{matrix} \right]
    %\left[\xi_{x},\,\xi_{y}\right]^\top
    \sim \mathcal{N}_d \left( \boldsymbol{0}, \mathbf{R} \right)  \,\,.
\]
Moreover, because each $\xi_{i}$ is a monotonic function of $x_i$ $i=1,...,d$,
the ranks of all $x_i$ and of $\xi_{i}$ coincide; thus, also their Spearman correlation. 
Thanks to the result from \citet[Ch.9, points (9.15, 10.15)]{kendall1990rank}, we prove the thesis
$\hfill\square$
\fi

%\paragraph{Proof of Proposition \ref{prop:speraman_rho_copula}}

%\[
%\boldsymbol{\rho}_{i,j}=\frac{6}{\pi}\arcsin\left(\frac{\mathrm{R}_{i,j}}{2}\right)
%\]
%$\hfill\square$
%\smallskip
\noindent
\paragraph{Proof of Proposition \ref{cor:consis_with_clv}}
The consistency of $\widehat{\boldsymbol{\varrho}}$ comes from the properties of IFM estimators \citep[see, e.g.][Sec.10.1]{joe1997multivariate} and of extemum estimators \citep[see, e.g.,][when the estimator $\widehat{\boldsymbol{\varrho}}$ does not satisfy the first order conditions]{newey1994large}.\\
To prove the consistency of $\widehat{\rho}_{ij}$ $\forall\,i,j=1,\ldots,d$ $i\neq j$, let us introduce the empirical copula \citep[][]{deheuvels1979fonction}
\begin{equation*}
    C_{N,ij}(u_i,u_j):=\frac{1}{N}\sum_{t=1}^N\mathbbm{1}_{\left\{\frac{Rank\left[x_i^{(t)}\right]}{N+1}\leq u_i\,;\,\frac{Rank\left[x_j^{(t)}\right]}{N+1}\leq u_j\right\}}\,\,,
\end{equation*}
where the rank is defined as
\[
Rank\left[x_i^{(t)}\right]=\sum_{k=1}^N\mathbbm{1}_{\left\{x_i^{(k)}\leq x_i^{(t)}\right\}}\,\,.
\]
It's straightforward to prove that
\begin{equation*}
    \widehat{\rho}_{ij}=\frac{N+1}{N-1}\left(12\int_0^1\int_0^1C_{N,ij}(u_i,u_j)\mathrm{d}u_i\mathrm{d}u_j-3\right)\,\,.
\end{equation*}
Thanks to Theorem 3.1 from \cite{deheuvels1979fonction}
\begin{equation*}
    \sup_{u_i,u_j\in[0,1]}\left|C_{N,ij}(u_i,u_j)-C_{ij}(u_i,u_j;\varrho_{ij})\right|\overset{\text{a.s}}{=}0\,\,,
\end{equation*}
where $C_{ij}(\bigcdot;\varrho_{ij})$ denotes the copula between the rvs $X_i$ and $X_j$.\\
Thus, considering the definition of Spearman's $\rho$ \citep[cf.,][Def.3.3, p.100]{CLV}, the Spearman correlation estimator $\rho_{ij}$ is consistent, i.e.,
\begin{equation*}
    \widehat{\rho}_{ij}\overset{\text{a.s.}}{\longrightarrow}\rho_{ij}=12\int_0^1\int_0^1C_{ij}(u_i,u_j;\varrho_{ij})\mathrm{d}u_i\mathrm{d}u_j-3\,\,,
\end{equation*}
thus, thanks to Proposition \ref{prop:speraman_rho_copula} below
\[
\phantom{\square} \hspace{4.5cm}
\displaystyle\widehat{\mathrm{R}}_{i,j}=2\sin\left(\frac{\pi}{6}\widehat{\rho}_{i,j}\right)\overset{\text{a.s.}}{\longrightarrow}2\sin\left(\frac{\pi}{6}\rho_{i,j}\right)=\mathrm{R}_{i,j}.
\hspace{4.5cm} \square
\]

\medskip
\section*{Auxiliary results}
In the following, we report two auxiliary results needed for the proofs in this Appendix.

\begin{proposition}\label{pr:uniformly_convergence}
    The log-likelihood $\llcb(\boldsymbol{\Theta},\boldsymbol{\varrho}\mid\mathcal{X})$ converges locally uniformly almost surely to $\llcbsp(\boldsymbol{\Theta},\boldsymbol{\varrho}\mid\mathcal{X})$, i.e.,
    \begin{equation*}
        \forall\,K\subset \mathcal{K}_d,\,K\;\text{compact}\quad\sup_{\boldsymbol{\varrho}\in K}\left|\llcbsp(\boldsymbol{\Theta},\boldsymbol{\varrho}\mid\mathcal{X})-\llcb(\boldsymbol{\Theta},\boldsymbol{\varrho}\mid\mathcal{X})\right|\overset{\text{a.s.}}{\longrightarrow}0\quad\text{as}\,\,\boldsymbol{p}\to\boldsymbol{1}\,\,.
    \end{equation*}
\end{proposition}
\begin{proof}
    For all $\boldsymbol{x}\in\mathcal{X}$, we can write the per-observation log-likelihood as
    \[
    \ellcp(\boldsymbol{\Theta},\boldsymbol{\varrho}\mid\boldsymbol{x})=\sum_{I \subseteq \mathbb{S}} \mathbbm{1}_{\{I=\mathcal{S}(\boldsymbol{x})\}} \, 
    \left(\sum_{i \in I} \ln \left[ p_i \psi_i(x_i;\boldsymbol{\theta}_i) \right]
            +
            \ln \left[
                \frac{\partial^{\, |I|} \,
                    C(u_1, \dots, u_d; \boldsymbol{\varrho})}
                {\prod_{i \in I} \, \partial u_i} \,
            \right]_{u_i = F_i(x_i;p_i)}\right)\,\,.
    \]
    Thus, for a sample $\mathcal{X}$, we can write
    \begin{equation}
    \label{eq:inequalities_sup}
    \begin{aligned}
        \sup_{\boldsymbol{\varrho}\in K}&\left|\llcbsp(\boldsymbol{\Theta},\boldsymbol{\varrho}\mid\boldsymbol{x})-\llcb(\boldsymbol{\Theta},\boldsymbol{\varrho}\mid\boldsymbol{x})\right|\leq\sum_{\boldsymbol{x}\in\mathcal{X}}\sup_{\boldsymbol{\varrho}\in K}\left|\ell(\boldsymbol{\Theta},\boldsymbol{\varrho}\mid\boldsymbol{x})-\ell_{\boldsymbol{p}}(\boldsymbol{\Theta},\boldsymbol{\varrho}\mid\boldsymbol{x})\right|\leq\\&\leq\sum_{\boldsymbol{x}\in\mathcal{X}}\sup_{\boldsymbol{\varrho}\in K}\left|\ell(\boldsymbol{\Theta},\boldsymbol{\varrho}\mid\boldsymbol{x})-\mathbbm{1}_{\{\mathbb{S}=\boldsymbol{\mathcal{S}}(\boldsymbol{x})\}}\ell_{\boldsymbol{p}}^\mathbb{S}(\boldsymbol{\Theta},\boldsymbol{\varrho}\mid\boldsymbol{x})\right|+\sum_{\boldsymbol{x}\in\mathcal{X}}\sup_{\boldsymbol{\varrho}\in K}\left|\sum_{I\subset \mathbb{S}}\mathbbm{1}_{\{I=\boldsymbol{\mathcal{S}}(\boldsymbol{x})\}}\ell_{\boldsymbol{p}}^I(\boldsymbol{\Theta},\boldsymbol{\varrho}\mid\boldsymbol{x})\right|\,\,,
            \end{aligned}
    \end{equation}
    where
    \begin{equation*}
        \ell_{\boldsymbol{p}}^I(\boldsymbol{\Theta},\boldsymbol{\varrho}\mid\boldsymbol{x}):=\sum_{i \in I} \ln \left[ p_i \psi_i(x_i;\boldsymbol{\theta}_i) \right]
            +
            \ln \left[
                \frac{\partial^{\, |I|} \,
                    C(u_1, \dots, u_d; \boldsymbol{\varrho})}
                {\prod_{i \in I} \, \partial u_i} \,
            \right]_{u_i = F_i(x_i;p_i)}\,\,.
    \end{equation*}
    The first addend in the last inequality in \eqref{eq:inequalities_sup} goes to $0$, indeed
    \begin{align*}
\sum_{\boldsymbol{x}\in\mathcal{X}}&\sup_{\boldsymbol{\varrho}\in K}\left|\sum_{I\subset \mathbb{S}}\mathbbm{1}_{\{I=\boldsymbol{\mathcal{S}}(\boldsymbol{x})\}}\ell_{\boldsymbol{p}}^I(\boldsymbol{\Theta},\boldsymbol{\varrho}\mid\boldsymbol{x})\right|\leq\sum_{\boldsymbol{x}\in\mathcal{X}}\sum_{I\subset \mathbb{S}}\sup_{\boldsymbol{\varrho}\in K}\left|\mathbbm{1}_{\{I=\boldsymbol{\mathcal{S}}(\boldsymbol{x})\}}\ell_{\boldsymbol{p}}^I(\boldsymbol{\Theta},\boldsymbol{\varrho}\mid\boldsymbol{x})\right|=\\
&=\sum_{\boldsymbol{x}\in\mathcal{X}}\sum_{I\subset \mathbb{S}}\sup_{\boldsymbol{\varrho}\in K}\left|\frac{\mathbbm{1}_{\{I=\boldsymbol{\mathcal{S}}(\boldsymbol{x})\}}}{p_I^\perp}\,p_I^\perp\,\ell_{\boldsymbol{p}}^I(\boldsymbol{\Theta},\boldsymbol{\varrho}\mid\boldsymbol{x})\right|\leq\sum_{\boldsymbol{x}\in\mathcal{X}}\sum_{I\subset \mathbb{S}}\sup_{\boldsymbol{\varrho}\in K}\left|\frac{\mathbbm{1}_{\{I=\boldsymbol{\mathcal{S}}(\boldsymbol{x})\}}}{p_I^\perp}\right|\,\sup_{\boldsymbol{\varrho}\in K}\left|p_I^\perp\right|\,\sup_{\boldsymbol{\varrho}\in K}\left|\ell_{\boldsymbol{p}}^I(\boldsymbol{\Theta},\boldsymbol{\varrho}\mid\boldsymbol{x})\right|\,\,,
    \end{align*}
    It's straightforward to prove that
    \[
    p_I^\perp\leq 1-p_i\quad\forall\,i\in I\quad\forall\,\boldsymbol{\varrho}\in K\,\,,
    \]
    thus,
    \[
    \sup_{\boldsymbol{\varrho}\in K}\,p_I^\perp\leq 1-p_i\quad\forall\,i\in I\,\,,
    \]
    this implies that
    \begin{equation}
    \label{eq:unif_conv_addend}
    \sum_{\boldsymbol{x}\in\mathcal{X}}\sum_{I\subset \mathbb{S}}\sup_{\boldsymbol{\varrho}\in K}\left|\frac{\mathbbm{1}_{\{I=\boldsymbol{\mathcal{S}}(\boldsymbol{x})\}}}{p_I^\perp}\right|\,\sup_{\boldsymbol{\varrho}\in K}\left|p_I^\perp\right|\,\sup_{\boldsymbol{\varrho}\in K}\left|\ell_{\boldsymbol{p}}^I(\boldsymbol{\Theta},\boldsymbol{\varrho}\mid\boldsymbol{x})\right|\overset{\text{a.s.}}{\longrightarrow}0
    \end{equation}
    for $\boldsymbol{p}\to\boldsymbol{1}$, since
    \begin{equation*}
        \left|\frac{\mathbbm{1}_{\{I=\boldsymbol{\mathcal{S}}(\boldsymbol{x})\}}}{p_I^\perp}\right|\quad\text{and} \quad\left|\ell_{\boldsymbol{p}}^I(\boldsymbol{\Theta},\boldsymbol{\varrho}\mid\boldsymbol{x})\right|\,\,,
    \end{equation*}
    are bounded functions of $\boldsymbol{\varrho}\in K$.\footnote{The functions $1/p_I^\perp$ and $\ell_{\boldsymbol{p}}^I(\boldsymbol{\Theta},\boldsymbol{\varrho}\mid\boldsymbol{x})$ are continuous functions of $\boldsymbol{\varrho}\in K$, thus, for the Weierstrass Theorem, they assume a (finite) maximum over the compact set $K\subset\mathcal{K}_d$ almost surely. The function $\displaystyle\mathbbm{1}_{\{I=\boldsymbol{S}(\boldsymbol{x})\}}$ is surely bounded.}\\
    Let us now prove that the first addend in \eqref{eq:inequalities_sup} goes a.s. to $0$ as $\boldsymbol{p}\to\boldsymbol{1}$, indeed
    \begin{align*}
    \sup_{\boldsymbol{\varrho}\in K}\left|\ell(\boldsymbol{\Theta},\boldsymbol{\varrho}\mid\boldsymbol{x})-\mathbbm{1}_{\{\mathbb{S}=\boldsymbol{\mathcal{S}}(\boldsymbol{x})\}}\ell_{\boldsymbol{p}}^\mathbb{S}(\boldsymbol{\Theta},\boldsymbol{\varrho}\mid\boldsymbol{x})\right|&\leq\sup_{\boldsymbol{\varrho}\in K}\left|\ell(\boldsymbol{\Theta},\boldsymbol{\varrho}\mid\boldsymbol{x})(1-\mathbbm{1}_{\{\mathbb{S}=\boldsymbol{\mathcal{S}}(\boldsymbol{x})\}})\right|+\\&+\sup_{\boldsymbol{\varrho}\in K}\left|\mathbbm{1}_{\{\mathcal{S}=\boldsymbol{\mathbb{S}}(\boldsymbol{x})\}}(\ell(\boldsymbol{\Theta},\boldsymbol{\varrho}\mid\boldsymbol{x})-\ell_{\boldsymbol{p}}^\mathbb{S}(\boldsymbol{\Theta},\boldsymbol{\varrho}\mid\boldsymbol{x}))\right|\,\,.
    \end{align*}
    First, we observe that
    \begin{equation*}
        \sup_{\boldsymbol{\varrho}\in K}\left|\ell(\boldsymbol{\Theta},\boldsymbol{\varrho}\mid\boldsymbol{x})(1-\mathbbm{1}_{\{\mathbb{S}=\boldsymbol{\mathcal{S}}(\boldsymbol{x})\}})\right|\leq\sup_{\boldsymbol{\varrho}\in K}\left|\ell(\boldsymbol{\Theta},\boldsymbol{\varrho}\mid\boldsymbol{x})\right|\sup_{\boldsymbol{\varrho}\in K}\left|(1-\mathbbm{1}_{\{\mathbb{S}=\boldsymbol{\mathcal{S}}(\boldsymbol{x})\}})\right|\,\,,
    \end{equation*}
    and, similarly to \eqref{eq:unif_conv_addend}
    \[
    \sup_{\boldsymbol{\varrho}\in K}\left|1-\mathbbm{1}_{\{\mathbb{S}=\boldsymbol{\mathcal{S}}(\boldsymbol{x})\}}\right|=\sup_{\boldsymbol{\varrho}\in K}\left|\sum_{I\subset\mathbb{S}}\mathbbm{1}_{\{I=\boldsymbol{\mathcal{S}}(\boldsymbol{x})\}}\right|\leq\sum_{I\subset\mathbb{S}}\sup_{\boldsymbol{\varrho}\in K}\left|\frac{\mathbbm{1}_{\{I=\boldsymbol{\mathcal{S}}(\boldsymbol{x})\}}}{p_I^\perp}\right|\sup_{\boldsymbol{\varrho}\in K}\left|p_I^\perp\right|\overset{\text{a.s.}}{\longrightarrow}0
    \]
    uniformly in $\boldsymbol{\varrho}\in K$, for $\boldsymbol{p}\to\boldsymbol{1}$. Moreover,
    \[
    \sup_{\boldsymbol{\varrho}\in K}\left|\mathbbm{1}_{\{\mathbb{S}=\boldsymbol{\mathcal{S}}(\boldsymbol{x})\}}(\ell(\boldsymbol{\Theta},\boldsymbol{\varrho}\mid\boldsymbol{x})-\ell_{\boldsymbol{p}}^\mathbb{S}(\boldsymbol{\Theta},\boldsymbol{\varrho}\mid\boldsymbol{x}))\right|\leq    \sup_{\boldsymbol{\varrho}\in K}\left|\mathbbm{1}_{\{\mathbb{S}=\boldsymbol{\mathcal{S}}(\boldsymbol{x})\}}\right|\,\sup_{\boldsymbol{\varrho}\in K}\left|\ell(\boldsymbol{\Theta},\boldsymbol{\varrho}\mid\boldsymbol{x})-\ell_{\boldsymbol{p}}^\mathbb{S}(\boldsymbol{\Theta},\boldsymbol{\varrho}\mid\boldsymbol{x})\right|\,\,,
    \]
    and this last term goes to $0$ a.s., indeed,
    \begin{align*}
            \sup_{\boldsymbol{\varrho}\in K}\left|\ell(\boldsymbol{\Theta},\boldsymbol{\varrho}\mid\boldsymbol{x})-\ell_{\boldsymbol{p}}^\mathbb{S}(\boldsymbol{\Theta},\boldsymbol{\varrho}\mid\boldsymbol{x})\right|\leq\sup_{\boldsymbol{\varrho}\in K}\left(\lambda_{\max}\left(\mathbf{R}^{-1}\right)+1\right)&\,\left|\left|\Phi^{-1}(\Psi(\boldsymbol{x});\boldsymbol{\theta})-\Phi^{-1}(F(\boldsymbol{x};\boldsymbol{p}))\right|\right|_2\times\\&\times\left|\left|\Phi^{-1}(\Psi(\boldsymbol{x});\boldsymbol{\theta})+\Phi^{-1}(F(\boldsymbol{x};\boldsymbol{p}))\right|\right|_2\overset{\text{a.s.}}{\longrightarrow}0\,\,,
    \end{align*}
    because
    \[
        \left|\left|\Phi^{-1}(\Psi(\boldsymbol{x});\boldsymbol{\theta})-\Phi^{-1}(F(\boldsymbol{x};\boldsymbol{p}))\right|\right|_2\overset{\text{a.s.}}{\longrightarrow}0
    \]
    and
    \begin{equation*}
        \mathbb{P}\left(\lim_{\boldsymbol{p}\to\boldsymbol{1}}\left|\left|\Phi^{-1}(\Psi(\boldsymbol{x});\boldsymbol{\theta})+\Phi^{-1}(F(\boldsymbol{x};\boldsymbol{p}))\right|\right|_2<+\infty\right)=1
    \end{equation*}
    moreover, the map $\mathbf{R}\mapsto\left(\lambda_{\max}\left(\mathbf{R}^{-1}\right)+1\right)$ is a continuous map, thus, for the Weierstrass Theorem, $\exists\,\Lambda>0$ and finite such that
    \[
        \lambda_{\max}\left(\mathbf{R}^{-1}\right)+1\leq \Lambda\quad\forall\,\mathbf{R}\in K\text{ compact}
        \qedhere
    \]
\end{proof}

\begin{proposition}\label{prop:speraman_rho_copula}
Let $X_1, \dots, X_d$ be r.v.s with continuous marginals $\Psi_i(\bigcdot;\boldsymbol{\theta}_i)$, and suppose their dependence structure is given by a Gaussian copula with correlation matrix $\mathbf{R}$. 
Then, for each pair $X_i$ and $X_j$, the Spearman rank correlation $\rho_{ij}$ satisfies
\[
\rho_{ij} = \frac{6}{\pi} \arcsin \left( \frac{\mathrm{R}_{i,j}}{2} \right)
\quad\quad\quad \forall\, i,j = 1,\dots,d \,.
\]
\end{proposition}
\begin{proof}
    From the definition of Spearman's $\rho_{ij}$ \citep[cf.,][Def.3.3, p.100]{CLV}
\begin{equation}
    \label{eq:def_spearman_CLV}
    \rho_{ij}:=12\iint_{[0,1]^2}u_iu_jc_{i,j}(u_i,u_j;\mathrm{R}_{i,j})\mathrm{d}u_i\mathrm{d}u_j-3\,\,,
\end{equation}
where $c_{ij}(\bigcdot;\mathrm{R}_{i,j})$ denotes the copula density between the rvs $X_i$ and $X_j$.\\
With the change of variables $v_i=\Phi^{-1}(u_i)$, $v_j=\Phi^{-1}(u_j)$, equation \eqref{eq:def_spearman_CLV} becomes
\begin{equation*}
\rho_{ij}:=12\iint_{[0,1]^2}\Phi(v_i)\Phi(v_j)\varphi_2(v_i,v_j;\mathrm{R}_{i,j})\mathrm{d}u_i\mathrm{d}u_j-3=12\;\mathcal{I}(\mathrm{R}_{i,j})-3\,\,,
\end{equation*}
with the following definition of $\mathcal{I}(R_{i,j})$,
\[
\mathcal{I}(\mathrm{R}_{i,j}):=\iint_{[0,1]^2}\Phi(v_i)\Phi(v_j)\varphi_2(v_i,v_j;\mathrm{R}_{i,j})\mathrm{d}u_i\mathrm{d}u_j\,\,.
\]
If we take the derivative of $\mathcal{I}(\mathrm{R}_{i,j})$ w.r.t.\ $\mathrm{R}_{i,j}$ we get that
\begin{equation}
\label{eq:derivative_spearman_rho}
    \frac{\mathrm{d}\,\mathcal{I}(\mathrm{R}_{i,j})}{\mathrm{d}\,\mathrm{R}_{i,j}}=\frac{1}{2\pi\sqrt{4-\mathrm{R}_{i,j}^2}}\,\,.
\end{equation}
By integrating \eqref{eq:derivative_spearman_rho} w.r.t.\ $\mathrm{R}_{i,j}$, and imposing that $\rho_{i,j}=0$ if $\mathrm{R}_{i,j}=0$, we get the thesis.
\end{proof}

%\clearpage
%\newpage

% Reset counters for lemmas and equations
\setcounter{lemma}{0}
\setcounter{equation}{0}

\section*{Notation and shorthands}	

\begin{tabular} {|c|l|}
		\toprule
		\textbf{Symbol}& \textbf{Description}\\ \bottomrule
		$X_i$ & marginal rv in Comb-Bernoulli \\
        $\boldsymbol{X}$& vector $[X_1,\dots,X_d]^\top$\\
        $\boldsymbol{x}$ & realization of $\mathbf{X}$\\
        $d$ & dimension of Comb-Bernoulli model \\
        $\mathbb{S}$ & set of indices $\{1,\dots,d\}$ \\
        $p_i$ &  probability of the event $\{X_i>0\}$\\        
        $\boldsymbol{p}$ & vector of probabilities $p_i$ \\          $\lambda_i$ &  frequency in continuous-time model \\     
        $\boldsymbol{\lambda}$ & vector of intensities $\lambda_i$ \\
        $f_i(\bigcdot;p_i), F_i(\bigcdot;p_i)$ & pdf and cdf of $X_i$ \\
		$\boldsymbol{\theta}_i$ & set of parameters of the distribution of $X_i|X_i>0$\\ 
        $\boldsymbol{\Theta}$ & matrix $[\boldsymbol{\theta}_1,\ldots,\boldsymbol{\theta}_d]$ of parameters for the marginal severity distributions \\ 
        $\psi  ( \bigcdot;\boldsymbol{\theta}_i ), {\Psi} ( \bigcdot;\boldsymbol{\theta}_i )$ &  generic pdf and cdf of the marginal distribution of $X_i|X_i>0$\\
        $C(\bigcdot;\boldsymbol{\varrho})$ & copula with dependence parameters $\boldsymbol{\varrho}$ \\
		$c(\bigcdot;\boldsymbol{\varrho})$ & copula density \\
        $\overline{C}(\bigcdot;\boldsymbol{\varrho})$ & survival copula \\
        $\mathfrak{C}$ & L\'evy copula \\
        $\mathfrak{C}_{I}$ & L\'evy copula restricted to set of indices $I\subseteq\mathbb{S}$ \\
        $\mathcal{S}(\boldsymbol{x})$ & active set of $\boldsymbol{x}$, $\{i\in\{1,\dots,d\}:x_i>0\}$ \\
        $\mathsf{s}(\boldsymbol{x})$ & cardinality of the active set $\mathcal{S}(\boldsymbol{x})$ \\
        $\ell_{\boldsymbol{p}}(\boldsymbol{\Theta},\boldsymbol{\varrho}\mid\boldsymbol{x})$ & per-observation Comb-Bernoulli likelihood \\
        $\lcb(\boldsymbol{\Theta},\boldsymbol{\varrho}\mid\mathcal{X})$ & Comb-Bernoulli likelihood \\
        $\llcb(\boldsymbol{\Theta},\boldsymbol{\varrho}\mid\mathcal{X})$ & Comb-Bernoulli log-likelihood \\
        $\lcp(\boldsymbol{\Theta},\boldsymbol{\varrho}\mid\mathcal{X})$ & L\'evy copula model likelihood \\
        $\llcp(\boldsymbol{\Theta},\boldsymbol{\varrho}\mid\mathcal{X})$ & L\'evy copula model log-likelihood \\
        $\llcbsp(\boldsymbol{\Theta},\boldsymbol{\varrho}\mid\mathcal{X})$ & Comb-Bernoulli log-likelihood in the full time series limit $(\boldsymbol{p} \to \boldsymbol{1})$ \\
        $\ell(\boldsymbol{\Theta},\boldsymbol{\varrho}\mid\boldsymbol{x})$ & per-observation Comb-Bernoulli log-likelihood in the full time series limit \\        $\mathcal{X}$ & time series $[\boldsymbol{x}^{(t)}]_{t=1}^N$ \\
        $N$ & time series length \\
		$t$ & time series index \\
        $n_i$ & number of positive realizations of the rv $X_i$\\
        $n_I^\perp$ & number of realizations of the event $\{X_i>0\;\forall\,i\in I\}\,\cap\,\{X_j=0\;\forall\,j\notin I\}$\\
        $p_I^\perp$ & probability of the event $\{X_i>0\;\forall\,i\in I\}\,\cap\,\{X_j=0\;\forall\,j\notin I\}$\\
        $\mathcal{K}_d$ & set of $d\times d$ correlation matrices \\
        ${\varphi} ( \bigcdot ), {\Phi} ( \bigcdot )$ & pdf and cdf of a standard normal rv \\
        $\boldsymbol{\varrho}$ & copula dependence parameters \\
		\bottomrule		
\end{tabular}

\bigskip

\begin{tabular} {|c|l|}
		\toprule
		\textbf{Symbol}& \textbf{Description}\\ \bottomrule
        $\varphi_d(\bigcdot;\boldsymbol{\Sigma}),\,\Phi_d(\bigcdot;\boldsymbol{\Sigma})$ & pdf and cdf of a $d$-variate st.n. distribution with covariance matrix $\boldsymbol{\Sigma}$ \\
		$\boldsymbol{Y}\sim\Phi(\boldsymbol{\mu},\boldsymbol{\Sigma})$ & vector of rvs $\boldsymbol{Y}$ distributed as a normal with mean $\boldsymbol{\mu}$ and covariance $\boldsymbol{\Sigma}$ \\
		$\rhocb$ & dependence parameter estimator from the log-likelihood $\llcb(\boldsymbol{\Theta},\boldsymbol{\varrho}\mid\mathcal{X})$ \\
		$\rhocp$ & dependence parameter estimator from the log-likelihood $\llcp(\boldsymbol{\Theta},\boldsymbol{\varrho}\mid\mathcal{X})$ \\
        $\widehat{\boldsymbol{\varrho}}$ & dependence parameter estimator from the log-likelihood $\llcbsp(\boldsymbol{\Theta},\boldsymbol{\varrho}\mid\mathcal{X})$ \\
        $\boldsymbol{R}$ & Gaussian copula correlation \\
        $\boldsymbol{\rho}$ & Spearman's correlation\\
        $\widehat{\boldsymbol{\rho}}$ & Spearman's correlation estimator\\
        %$\Lambda_i$ & lower tail dependence coefficient \\
        $\mathbbm{1}$ & indicator function \\
        $|A|$ & determinant of a matrix $A$ \\
        $|I|$ & cardinality of a set $I$ \\
        $\nabla_{\boldsymbol{\varrho}}$ & gradient w.r.t.\ the components of $\boldsymbol{\varrho}$ \\
        %$\overset{\text{a.s.}}{=}$ & almost surely equality \\
		$\xrightarrow{\text{a.s.}}$ & almost surely convergence \\
		$\mathcal{P}$ & multinomial distribution \\
		$\delta_0(x)$ & Dirac delta function, centred in $0$ \\
		\bottomrule		
\end{tabular}

\bigskip

\begin{tabular}{|c|l|}
		\toprule
		\textbf{Shorthand}& \textbf{Description}\\ \bottomrule
		a.s.	& almost surely \\
		cdf 	& cumulative distribution function \\
		CI 	& Confidence Interval \\
        IFM 	& Inference Functions for Margins method \\
		ML 	& Maximum Likelihood \\
		%MLE 	& Maximum Likelihood Estimator \\
		pdf 	& probability density function \\
		rv 	& random variable \\
		st.n. 	& standard normal (random variable) \\
		w.r.t.\  	& with respect to \\
		\bottomrule
\end{tabular}

\end{document}